\documentclass[a4paper,11pt]{article}
\pdfoutput=1

\usepackage{jheppub}
\usepackage[T1]{fontenc}
\usepackage{amssymb,amsthm}
\usepackage{tikz}
\usepackage{multirow}
\usepackage{algorithm,algpseudocode}
\usepackage{afterpage}

\usetikzlibrary{decorations.markings}
\tikzset{->-/.style={
    decoration={
        markings,
        mark=at position 0.5 with {\arrow{>}}
    },
    postaction={decorate}
}}

\newtheorem{lemma}{Lemma}
\newcommand{\U}{\operatorname{U}}
\newcommand{\env}{\textbf{env}}

\title{Breeding realistic D-brane models}

\author{Gregory J.\ Loges}
\author{and Gary Shiu}
\affiliation{Department of Physics, University of Wisconsin-Madison,\\ 1150 University Ave, Madison WI 53706, U.S.A.}

\emailAdd{gloges@wisc.edu}
\emailAdd{shiu@physics.wisc.edu}

\abstract{
Intersecting branes provide a useful mechanism to construct particle physics models from string theory with a wide variety of desirable characteristics. The landscape of such models can be enormous, and navigating towards regions which are most phenomenologically interesting is potentially challenging. Machine learning techniques can be used to efficiently construct large numbers of consistent and phenomenologically desirable models. In this work we phrase the problem of finding consistent intersecting D-brane models in terms of genetic algorithms, which mimic natural selection to evolve a population collectively towards optimal solutions. For a four-dimensional ${\cal N}=1$ supersymmetric type IIA orientifold with intersecting D6-branes, we demonstrate that $\mathcal{O}(10^6)$ unique, fully consistent models can be easily constructed, and, by a judicious choice of search environment and hyper-parameters, $\mathcal{O}(30\%)$ of the found models contain the desired Standard Model gauge group factor. Having a sizable sample allows us to draw some preliminary landscape statistics of intersecting brane models both with and without the restriction of having the Standard Model gauge factor.
}

\begin{document} 

\maketitle
\flushbottom

\section{Introduction}
The landscape of string vacua is famously enormous. Estimates for the number of vacua in type IIB flux compactifications easily reach $10^{500}$ \cite{Ashok:2003gk} and Calabi-Yau four-fold compactifications of F-theory lead possibly to upwards of $10^{272,000}$ flux vacua \cite{Taylor:2015xtz}. Such large ensembles lend themselves well to statistical arguments \cite{Douglas:2003um,Gmeiner:2005vz}, but there remains the question of whether various desirable properties of the vacua are sufficiently independent so as to make statistical extrapolations robust. While some observables are expected to be statistically independent, certain correlations have been observed in specific searches.\footnote{An example is the correlation between the string coupling $g_s$ and flux superpotential $W_0$ observed in \cite{Cole:2019enn}. It remains to be seen whether such correlation is due to sampling bias of the search algorithms used.} Another issue faced when confronted with such large ensembles is how to sample efficiently or navigate towards those corners with uncommon or desirable properties such as phenomenologically viable cosmological constant or particle spectrum.

In recent years data science techniques such as topological data analysis and machine learning have proven invaluable in understanding the underlying structure of the string landscape~\cite{Cirafici2015,Carifio:2017bov,Mutter:2018sra,Cole:2018emh,Cole:2019enn,Deen:2020dlf,Otsuka:2020nsk,Krippendorf:2021uxu,Constantin:2021for,Cole:2021nnt,Gao:2021xbs,Berman2021}
or simply tackling the enormous (theoretical) datasets involved~\cite{Abel2014,He:2017aed,Ruehle:2017mzq,Krefl:2017yox,Bull:2018uow,Klaewer:2018sfl,Bull:2019cij,Halverson:2019tkf,Bena:2021wyr,Abel:2021ddu,Faraggi2021}.
One obstacle that one can be faced with is the fact that the conditions which describe vacua are systems of Diophantine equations (i.e.~one is looking for solutions in the integers), which are the subject of Hilbert's tenth problem~\cite{Hilbert}. The task of determining if an arbitrary Diophantine equation has solutions was shown to be an undecidable problem by Matiyasevich in 1970~\cite{MR0258744}, which already hints at the rich structure of such equations. Of course, for any given Diophantine equation it may be relatively easy to find one or more integer solutions. A more prudent question in such cases where solutions are known to exist is how many solutions exist with integers below some bound. How quickly does the number of solutions grow as the bound is increased? Are there only finitely many solutions?

When traditional gradient-descent-like methods are unavailable because of the discreteness of the variables involved, genetic algorithms provide an alternative method to search for solutions. As the name suggests, genetic algorithms take inspiration from evolving biological systems to solve some problem via ``natural selection''. A population of individuals becomes fitter as cross-over (i.e.~``breeding'') and random mutations improve offspring over the generations. If the goal is to optimize some fitness function quantifying how well the problem is solved, heuristically the cross-over provides large changes which help to avoid local extrema, while mutations can provide smaller changes that fine-tune the optimization.

In this paper we initiate a program of using genetic algorithms to breed realistic D-brane models. The concrete setup we consider is type IIA orientifold with intersecting branes, which has proven to be a fruitful framework for constructing realistic particle physics models from string theory (see \cite{Lust2004,Blumenhagen2005,Blumenhagen:2006ci,ibanez2012} for reviews). Among the desirable properties built-in in this approach is the existence of chiral matter at brane intersections \cite{Berkooz:1996km}. To learn how genetic algorithms find optimal solutions, it is advantageous to start with the simplest example. The ``harmonic oscillator'' in this context is a $\mathbb{T}^6/(\mathbb{Z}_2\times\mathbb{Z}_2)$ orientifold where the consistency conditions for four-dimensional ${\cal N}=1$ supersymmetric intersecting brane models were fully developed and three-family Standard-like models were first constructed~\cite{Cvetic:2001nr,Cvetic:2001tj}. As will be described below, the set of consistency conditions on brane orientations constitute a system of Diophantine equations and inequalities. In addition, the search space is high-dimensional: for a fixed number of D6-brane stacks, $k$, one needs to specify $7k$ integers and a na\"ive brute-force search up to numbers of size $L$ amounts to checking $\mathcal{O}(L^{7k})$ combinations. Even for $k\lesssim 5$ this quickly becomes intractable. We explore the viability of using genetic algorithms to efficiently search for consistent intersecting brane models.

The remainder of this paper is organized as follows. In Sec.~\ref{sec:review} we review the model of interest and its consistency conditions. In Sec.~\ref{sec:GA} we describe the genetic algorithm used to search for consistent intersecting brane models, including choices for the chromosome structure, cross-over methods, mutations and fitness function. In Sec.~\ref{sec:results} we present an interpretation of the strategy employed by the genetic algorithm, as well as results for three different ensembles of models. Finally, we conclude in Sec.~\ref{sec:discussion} with a discussion and thoughts for future work. Some details of the genetic algorithm are relegated to App.~\ref{app:GA}.


\section{\texorpdfstring{\boldmath Review of $\mathbb{T}^6/(\mathbb{Z}_2\times\mathbb{Z}_2)$ orientifold}{Review of T6/Z2xZ2 orientifold}}
\label{sec:review}

In this section we review the features of type IIA orientifold models relevant to our study using genetic algorithms. For a more comprehensive discussion we direct the interested reader to~\cite{Lust2004,Blumenhagen2005,Blumenhagen:2006ci,ibanez2012}. We restrict attention to quantities that are topological in nature and can be computed reliably without details of the compactification geometry. This includes the consistency conditions which ensure the theory is free of anomalies and preserves 4D $\mathcal{N}=1$ supersymmetry.

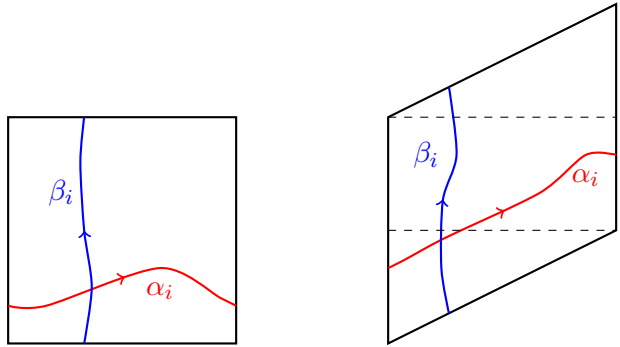
\begin{figure}[t]
	\centering
	\begin{tikzpicture}
		\draw[thick] (0,0) -- (3,0) -- (3,3) -- (0,3) -- cycle;
		\draw[thick,red,->-] plot [smooth] coordinates {(0,0.5) (0.5,0.5) (2,1) (2.8,0.6) (3,0.5)};
		\node[red] at (2,0.7) {$\alpha_i$};
		\draw[thick,blue,->-] plot [smooth] coordinates {(1,0) (1.1,0.75) (1,1.5) (0.95,2.4) (1,3)};
		\node[blue] at (0.7,2) {$\beta_i$};
		
		\draw[thick] (5,0) -- (8,1.5) -- (8,4.5) -- (5,3) -- cycle;
		\draw[thick,red,->-] plot [smooth] coordinates {(5,1) (5.2,1.1) (5.75,1.4) (7,2) (7.6,2.5) (8,2.5)};
		\node[red] at (7.6,2.2) {$\alpha_i$};
		\draw[thick,blue,->-] plot [smooth] coordinates {(5.8,0.4) (5.7,1) (5.72,1.9) (5.9,2.5) (5.8,3.4)};
		\node[blue] at (5.5,2.5) {$\beta_i$};
		\draw[dashed] (5,1.5) -- (8,1.5);
		\draw[dashed] (5,3) -- (8,3);
	\end{tikzpicture}
	\caption{The two torus shapes (rectangular/untilted and tilted) compatible with the orientifold projection. The fundamental 1-cycles $\alpha_i$ and $\beta_i$ are shown in red and blue.}
	\label{fig:torishapes}
\end{figure}

Our starting point is Type IIA string theory on the orbifold $X^6=\mathbb{T}^6/(\mathbb{Z}_2\times\mathbb{Z}_2)$. Writing $z_i$ for the complex coordinates of $\mathbb{T}^6$, the group $\mathbb{Z}_2\times\mathbb{Z}_2$ is generated by
\begin{equation}
\begin{aligned}
	\theta:\quad (z_1,z_2,z_3) &\mapsto (z_1,-z_2,-z_3) \,,\\
	\omega:\quad (z_1,z_2,z_3) &\mapsto (-z_1,z_2,-z_3) \,.
\end{aligned}
\end{equation}
This projects out many moduli of the $\mathbb{T}^6$, in particular the 6-torus is now factorizable: $\mathbb{T}^6=\mathbb{T}^2\times\mathbb{T}^2\times\mathbb{T}^2$. Modding out by the orientifold action $\Omega\overline{\sigma}(-1)^F$, where $\Omega$ and $(-1)^F$ are world-sheet parity and fermion number, respectively, and $\overline{\sigma}$ acts by complex conjugation on each $\mathbb{T}^2$, introduces O6-planes at the fixed points of the involution. These carry R-R charges that must be canceled by the introduction of D6-branes (along with their orientifold images). Writing $[\Pi_\text{O6}]$ for the total homology class of the O6-planes on $X^6$, $[\Pi_a]$ for the homology class of a stack of $N_a$ D6-branes and $[\Pi_{a'}]$ for its orientifold image, the R-R tadpole cancellation condition reads
\begin{equation}
\label{eq:homclassvanish}
	\sum_a N_a\big([\Pi_a] + [\Pi_{a'}]\big) - 4[\Pi_\text{O6}] = 0 \,.
\end{equation}
It is convenient to expand all 3-cycles in terms of a symplectic basis of $H_3(X^6) = H_3^+(X^6)\oplus H_3^-(X^6)$, e.g.
\begin{equation}
\label{eq:homexpansion}
	{}[\Pi_a] = \sum_I \big(\widetilde{X}_a^I[\pi_I^+] + \widetilde{Y}_a^I[\pi_I^-]\big) \,,
\end{equation}
where $[\pi_I^+]\in H_3^+(X^6)$ and $[\pi_I^-]\in H_3^-(X^6)$ are orientifold-even and orientifold-odd, respectively. For each $\mathbb{T}^2$ factor of the 6-torus we can choose the complex coordinate so that the identifications take the form
\begin{equation}
	z_i \sim z_i + (1 + ib_i) \sim z_i + i \,.
\end{equation}
Compatibility with the orientifold action imposes $2b_i\in\mathbb{Z}$, and the only two unique choices are $b_i=0$ (rectangular/untilted) and $b_i=\frac{1}{2}$ (tilted): see Fig.~\ref{fig:torishapes} for a sketch of these two possibilities along with the fundamental 1-cycles of the torus, $\alpha_i$ and $\beta_i$. We restrict attention to factorizable 3-cycles, which are determined by three pairs of coprime winding numbers, $\{(n^i,m^i)\}_{i=1,2,3}$, around the fundamental torus cycles $\alpha_i,\beta_i$ on which the orientifold action acts as
\begin{equation}
\begin{aligned}
	\overline{\sigma}:\quad \alpha_i &\mapsto \alpha_i-2b_i\beta_i \,,\\
	\beta_i &\mapsto -\beta_i \,.
\end{aligned}
\end{equation}
A symplectic basis of factorizable orientifold-even and orientifold-odd 3-cycles can then be written as
\begin{equation}
\begin{aligned}
	{}[\pi_0^+] &= [\widetilde{\alpha}_1\widetilde{\alpha}_2\widetilde{\alpha}_3] \,, & [\pi_1^+] &= -[\widetilde{\alpha}_1\beta_2\beta_3] \,, & [\pi_2^+] &= -[\beta_1\widetilde{\alpha}_2\beta_3] \,, & [\pi_3^+] &= -[\beta_1\beta_2\widetilde{\alpha}_3] \,,\\
	[\pi_0^-] &= [\beta_1\beta_2\beta_3] \,, & [\pi_1^-] &= -[\beta_1\widetilde{\alpha}_2\widetilde{\alpha}_3] \,, & [\pi_2^-] &= -[\widetilde{\alpha}_1\beta_2\widetilde{\alpha}_3] \,, & [\pi_3^-] &= -[\widetilde{\alpha}_1\widetilde{\alpha}_2\beta_3] \,,
\end{aligned}
\end{equation}
where we have introduced the orientifold-even combinations $\widetilde{\alpha}_i = \alpha_i - b_i\beta_i$ along with a corresponding shift in the winding numbers to $\widetilde{m}^i = m^i + b_in^i$. The expansion coefficients of equation~\eqref{eq:homexpansion} are then
\begin{equation}
\begin{aligned}
	\widetilde{X}_a^0 &= n_a^1n_a^2n_a^3 \,, & \widetilde{X}_a^1 &= -n_a^1\widetilde{m}_a^2\widetilde{m}_a^3 \,, & \widetilde{X}_a^2 &= -\widetilde{m}_a^1n_a^2\widetilde{m}_a^3 \,, & \widetilde{X}_a^3 &= -\widetilde{m}_a^1\widetilde{m}_a^2n_a^3 \,,\\
	\widetilde{Y}_a^0 &= \widetilde{m}_a^1\widetilde{m}_a^2\widetilde{m}_a^3 \,, & \widetilde{Y}_a^1 &= -\widetilde{m}_a^1n_a^2n_a^3 \,, & \widetilde{Y}_a^2 &= -n_a^1\widetilde{m}_a^2n_a^3 \,, & \widetilde{Y}_a^3 &= -n_a^1n_a^2\widetilde{m}_a^3 \,.
\end{aligned}
\end{equation}
The coefficients for the orientifold image are found by the replacement $\widetilde{m}^i\to -\widetilde{m}^i$ (equivalently, $\widetilde{Y}^I\to-\widetilde{Y}^I$).

O6-planes are located at the fixed-point loci of $\Omega\overline{\sigma}(-1)^F$, $\Omega\overline{\sigma}(-1)^F\theta$, $\Omega\overline{\sigma}(-1)^F\omega$ and $\Omega\overline{\sigma}(-1)^F\theta\omega$; all-told the O6-plane homology class is
\begin{equation}
\label{eq:O6homclass}
	{}[\Pi_\text{O6}] = 4[\pi_0^+] + 4(1-b_2)(1-b_3)[\pi_1^+] + 4(1-b_1)(1-b_3)[\pi_2^+] + 4(1-b_1)(1-b_2)[\pi_3^+] \,,
\end{equation}
i.e.\ $\widetilde{X}_\text{O6}^0=4$, $\widetilde{X}_\text{O6}^1=4(1-b_2)(1-b_3)$, etc.\ and $\widetilde{Y}_\text{O6}^I=0$. To unify the treatment of rectangular and tilted tori, it is useful to introduce the scaled winding numbers $\widehat{m}^i = (1+2b_i)\widetilde{m}^i = m^i + 2b_i(n^i+m^i)\in\mathbb{Z}$ and coefficients $\widehat{X}^1 = -n^1\widehat{m}^2\widehat{m}^3\in\mathbb{Z}$, etc. For example, we have simply $\widehat{X}_\text{O6}^I=4$. With this normalization choice, the consistency conditions read (e.g.\ see~\cite{Lust2004,Blumenhagen2005,ibanez2012})
\begin{align}
\begin{aligned}
	\text{Tadpole:}& & \sum_a N_a\widehat{X}_a^I &= 8 \qquad & &\text{for each }I\in\{0,1,2,3\} \,,\\
	\text{K-theory:}& & \sum_a N_a\widehat{Y}_a^I &\in 2\mathbb{Z} & &\text{for each }I\in\{0,1,2,3\} \,,\\
	\text{SUSY:}& & \qquad \left\{\begin{array}{r}
		\sum\limits_{I=0}^3 \widehat{X}_a^I\widehat{U}_I > 0\\
		\sum\limits_{I=0}^3\frac{\widehat{Y}_a^I}{\widehat{U}_I} = 0
	\end{array}\right.\hspace{-24pt}& & &\text{for each stack }a \,,
\end{aligned}
\end{align}
where $\widehat{U}_0 = R_x^1R_x^2R_x^3$ and $\widehat{U}_1 = (1-b_2)(1-b_3)R_x^1R_y^2R_y^3$, etc.\ with $R_x^i,R_y^i>0$ the tori radii. The tadpole condition, as discussed above, amounts to requiring the total R-R charge on the compact space vanish. The K-theory condition avoids there being a global anomaly, and the SUSY conditions ensure that each D6-brane stack preserves the same supersymmetries as do the O6-planes. With $k$ D6-brane stacks, there are a total of $2k+8$ conditions to be satisfied.

\begin{table}[t]
	\centering
	\renewcommand{\arraystretch}{1.25}
	\begin{tabular}{|c||c|c|cc|c|}
		\hline
		Type & $\#(\widehat{X}^I,\widehat{Y}^I=0)$ & $\#(\widehat{X}^I>0)$ & SUSY-X & SUSY-Y & Example windings $(b_i=0)$ \\ \hline\hline
		$A$ & \multirow{2}{*}{$(0,0)$} & 3 & \checkmark & \checkmark & $(+,+,+,-,+,-)$\\ \cline{1-1}\cline{3-6}
		$A'$ & & 1 & \checkmark$^\ast$ & \checkmark$^\ast$ & $(+,-,+,-,+,-)$\\ \hline
		$B$ & \multirow{3}{*}{$(2,2)$} & 2 & \checkmark & \checkmark & $(\;1,\;0,+,+,+,-)$\\ \cline{1-1}\cline{3-6}
		\multirow{2}{*}{$B'$} & & $1$ & \checkmark & & $(\;0,\;1,+,-,-,-)$\\
		 & & $0$ & & \checkmark & $(\;1,\;0,+,+,-,+)$\\ \hline
		$C$ & \multirow{2}{*}{$(3,4)$} & 1 & \checkmark & \checkmark & $(\;1,\;0,\;1,\;0,\;1,\;0)$\\ \cline{1-1}\cline{3-6}
		$C'$ & & 0 & & \checkmark & $(\;1,\;0,\;0,\;1,\;0,\;1)$\\ \hline
		\multirow{2}{*}{$D'$} & \multirow{2}{*}{$(3,3)$} & 1 & \checkmark & & $(+,-,\;1,\;0,\;0,\;1)$\\
		 & & 0 & & & $(+,+,\;1,\;0,\;0,\;1)$\\ \hline
		$E'$ & $(4,3)$ & 0 & \checkmark & & $(\;0,\;1,\;1,\;0,\;1,\;0)$\\ \hline
	\end{tabular}
	\caption{Classification of branes according to how many $\widehat{X}^I,\widehat{Y}^I$ are positive, negative or zero. Only types $A$, $B$ and $C$ can satisfy both SUSY conditions. Primes indicate brane types that cannot satisfy both SUSY conditions: note that while type $A'$ branes \emph{can} satisfy either SUSY condition it is straightforward to prove that they cannot do so \emph{simultaneously} (e.g.\ see~\cite{Douglas2006}).}
	\label{tab:classification}
\end{table}

The moduli which appear in the SUSY conditions for each stack are \emph{a priori} real-valued. However, the following lemma is useful in restricting the possible values of $\widehat{U}_I$ which need to be considered.
\begin{lemma}
Without loss of generality, the moduli $\widehat{U}_I$ may be chosen to be a set of four positive, coprime integers.
\label{positive-coprime-lemma}
\end{lemma}
\begin{proof}
Write the second SUSY condition above as a matrix equation,
\begin{equation}
    \mathcal{Y}\mathcal{U} = 0 \,, \qquad [\mathcal{Y}]_{aI} = \widehat{Y}_a^I \,, \qquad  [\mathcal{U}]_I = (\widehat{U}_I)^{-1} \,.
\end{equation}
It suffices to show that if a solution exists with all $\widehat{U}_I$ positive, then there exist solutions where all $\widehat{U}_I$ are positive and rational. The $k\times 4$ integer matrix $\mathcal{Y}$ may be brought to \emph{Smith normal form}, $\mathcal{G}=\mathcal{SYT}$, where $\mathcal{S}$ and $\mathcal{T}$ are unimodular integer matrices of size $k\times k$ and $4\times 4$, respectively. The $k\times 4$ integer matrix $\mathcal{G}$ has nonzero entries only on its main diagonal. It is then clear that there exists a rational basis for solutions of $\mathcal{GV}=0$, from which a rational basis for solutions of $\mathcal{YU}=0$ may be found via $\mathcal{U}=\mathcal{TV}$, since $\mathcal{T}$ is an integer matrix and $\mathcal{YTV}=\mathcal{S}^{-1}\mathcal{GV}=0$.
\end{proof}
\noindent
We will use this fact later in Sec.~\ref{sec:results} when searching for consistent models.


\subsection{Brane classification}
\label{sec:braneclassification}

\begin{table}[t]
	\centering
	\begin{tabular}{|c|cc|}
		\hline
		$b_i$ & $w^i$ & $\widehat{X}^I$\\ \hline\hline
		\multirow{4}{*}{$(0,0,0)$} & $(\;\;\;1,\;\;\;0,\;\;\;1,\;\;\;0,\;\;\;1,\;\;\;0)$ & $(1,0,0,0)$\\
		& $(\;\;\;1,\;\;\;0,\;\;\;0,\;\;\;1,\;\;\;0,\,-1)$ & $(0,1,0,0)$\\
		& $(\;\;\;0,\;\;\;1,\;\;\;1,\;\;\;0,\;\;\;0,\,-1)$ & $(0,0,1,0)$\\
		& $(\;\;\;0,\;\;\;1,\;\;\;0,\;\;\;1,\,-1,\;\;\;0)$ & $(0,0,0,1)$\\ \hline
		\multirow{4}{*}{$(0,0,\tfrac{1}{2})$} & $(\;\;\;1,\;\;\;0,\;\;\;1,\;\;\;0,\;\;\;2,\,-1)$ & $(2,0,0,0)$\\
		& $(\;\;\;1,\;\;\;0,\;\;\;0,\;\;\;1,\;\;\;0,\,-1)$ & $(0,1,0,0)$\\
		& $(\;\;\;0,\;\;\;1,\;\;\;1,\;\;\;0,\;\;\;0,\,-1)$ & $(0,0,1,0)$\\
		& $(\;\;\;0,\;\;\;1,\;\;\;0,\;\;\;1,\,-2,\;\;\;1)$ & $(0,0,0,2)$\\ \hline
		\multirow{4}{*}{$(0,\tfrac{1}{2},\tfrac{1}{2})$} & $(\;\;\;1,\;\;\;0,\;\;\;2,\,-1,\;\;\;2,\,-1)$ & $(4,0,0,0)$\\
		& $(\;\;\;1,\;\;\;0,\;\;\;0,\;\;\;1,\;\;\;0,\,-1)$ & $(0,1,0,0)$\\
		& $(\;\;\;0,\;\;\;1,\;\;\;2,\,-1,\;\;\;0,\,-1)$ & $(0,0,2,0)$\\
		& $(\;\;\;0,\;\;\;1,\;\;\;0,\;\;\;1,\,-2,\;\;\;1)$ & $(0,0,0,2)$\\ \hline
		\multirow{4}{*}{$(\tfrac{1}{2},\tfrac{1}{2},\tfrac{1}{2})$} & $(\;\;\;2,\,-1,\;\;\;2,\,-1,\;\;\;2,\,-1)$ & $(8,0,0,0)$\\
		& $(\;\;\;2,\,-1,\;\;\;0,\;\;\;1,\;\;\;0,\,-1)$ & $(0,2,0,0)$\\
		& $(\;\;\;0,\;\;\;1,\;\;\;2,\,-1,\;\;\;0,\,-1)$ & $(0,0,2,0)$\\
		& $(\;\;\;0,\;\;\;1,\;\;\;0,\;\;\;1,\,-2,\;\;\;1)$ & $(0,0,0,2)$\\ \hline
	\end{tabular}
	\caption{Winding numbers and tadpole contributions of type $C$ branes on (un)tilted tori.}
	\label{tab:typeCwindings}
\end{table}

It is illuminating to classify branes by how many of $\widehat{X}^I,\widehat{Y}^I$ are positive, negative or zero; in Table~\ref{tab:classification} we list a classification of branes into eight types. A fully consistent model is necessarily comprised of only type $A$, $B$ and $C$ branes. That the only negative contributions to tadpoles come from type $A$ branes and the fact that these are necessarily accompanied by positive contributions to the three other tadpoles was a critical ingredient in the proof of \cite{Douglas2006} that the number of solutions to the above consistency conditions is finite. Type $C$ branes (a.k.a.~filler branes as coined in \cite{Cvetic:2002pj}) are distinguished in that they automatically satisfy the SUSY conditions and contribute (positively) to only one tadpole. The winding numbers and tadpole contributions for type $C$ branes on (un)tilted tori in our normalizations are given in Table~\ref{tab:typeCwindings}.


\subsection{Phenomenological properties}
\label{sec:pheno}
Also of interest are phenomenological properties of the resulting 4D effective theory. Both the gauge group and chiral spectrum may be found with relative ease. The gauge group living on a stack of $N_a$ D6-branes depends on its orientation relative to its orientifold image. Generically such a stack contributes a $\U(N_a)$ factor, but those stacks with $\widehat{Y}_a^I=0$ (types $C$,$C'$) contribute $\operatorname{USp}(N_a)$. In addition, most $\U(1)$ factors are anomalous, gaining a mass of order $gM_s$ where $g$ is the corresponding gauge coupling and $M_s$ the string scale. Those linear combinations which remain massless are given by the condition ($x_a\in\mathbb{R}$)
\begin{equation}
	\sum_a x_a\mathrm{U}(1)_a \quad\text{massless} \quad\Longleftrightarrow\quad \sum_a x_aN_a\widehat{Y}_a^I = 0 \,.
\end{equation}

Chiral matter is localized at brane intersections and fall into the representations listed in Table~\ref{tab:reps}. Replication of matter representations arises from stacks, their orientifold images and O6-planes intersecting multiple times on $X^6$, with the intersection numbers $I_{ab}$ being given by
\begin{equation}
\begin{aligned}
	I_{ab} &= [\Pi_a]\circ[\Pi_b] = \prod_{i=1}^3(n_a^im_b^i-m_a^in_b^i) = \prod_{i=1}^3(1-b_i)(n_a^i\widehat{m}_b^i-\widehat{m}_a^in_b^i)\\
	&= \prod_{i=1}^3(1-b_i) \times \sum_{I=0}^3\big(\widehat{X}_a^I\widehat{Y}_a^I - \widehat{Y}_a^I\widehat{X}_b^I\big) \,.
\end{aligned}
\end{equation}

\begin{table}[t]
	\renewcommand{\arraystretch}{1.35}
	\centering
	\begin{tabular}{|l|cccc|}
	    \hline
	    Representation & $(\mathbf{N}_a,\overline{\mathbf{N}}_b)$ & $(\mathbf{N}_a,\mathbf{N}_b)$ & $\square\!\square_a$ & \resizebox{\width}{1.8\height}{$\boxminus$}$_a$\\
	    Multiplicity & $I_{ab}$ & $I_{ab'}$ & $\frac{1}{2}(I_{aa'}-I_{a\text{O6}})$ & $\frac{1}{2}(I_{aa'}+I_{a\text{O6}})$\\ \hline
	\end{tabular}
	\caption{Chiral spectrum arising from intersections between brane stacks and their orientifold images.}
	\label{tab:reps}
\end{table}

In Sec.~\ref{sec:results} we will be looking for one particular four-stack realisation of the MSSM gauge group and chiral spectrum. If we label the four stacks by $a,b,c,d$ with gauge group $\U(3)_a\times\U(2)_b\times\U(1)_c\times\U(1)_d$, the number of families for the quarks and leptons are (e.g.\ see~\cite{Gmeiner:2005vz})
\begin{equation}
\begin{aligned}
    n_Q &= I_{ab} + I_{ab'} \,,\\
    n_u &= I_{a'c} + I_{a'd} \,,\\
    n_d &= I_{a'c'} + I_{a'd'} + \frac{1}{2}(I_{aa'} + I_{a\text{O6}}) \,,\\
    n_L &= I_{bc} + I_{bd} + I_{b'c} + I_{b'd} \,,\\
    n_e &= \frac{1}{2}(I_{cc'}-I_{c\text{O6}}) + \frac{1}{2}(I_{dd'}-I_{d\text{O6}}) + I_{cd'} \,.
\end{aligned}
\end{equation}
Weak hypercharge is given by the linear combination $Q_Y = \frac{1}{6}Q_a + \frac{1}{2}Q_c + \frac{1}{2}Q_d$ and remains massless if $\widehat{Y}_a^I + \widehat{Y}_c^I + \widehat{Y}_d^I = 0$. Coupling constants are given by
\begin{equation}
    \alpha_a \propto \sum_{I=0}^3\widehat{X}_a^I\widehat{U}_I > 0
\end{equation}
with $\alpha_s=\alpha_a$, $\alpha_w=\alpha_b$ and $(\alpha_Y)^{-1} = \frac{1}{6}(\alpha_a)^{-1} + \frac{1}{2}(\alpha_c)^{-1} + \frac{1}{2}(\alpha_d)^{-1}$. The constant of proportionality above depends on the complex structure moduli and string scale, but ratios of couplings are unambiguous. One can also consider the Weinberg angle, which is given by the rational combination
\begin{equation}
    \sin^2{\theta_\text{W}} = \frac{1}{1 + \frac{\alpha_w}{\alpha_Y}} \,.
\end{equation}


\subsection{Symmetries and winding number operations}
\label{sec:symsandoperations}
The encoding of a homology class as a tuple of six winding numbers is not unique: the sign of winding numbers is not uniquely fixed. For example, referring to the winding numbers $n^i,m^i$ of a stack collectively as $w^i$, the winding numbers
\begin{equation}
	w^i=(3,1,0,-1,-1,2) \quad \Longleftrightarrow \quad (-3,-1,0,-1,1,-2)
\end{equation}
describe the same homology class for any $b_i$. In general the $\widehat{X}^I$ and $\widehat{Y}^I$ are unaltered if the signs of two $(n,m)$ pairs are changed simultaneously; the map from winding numbers to homology class is 4-to-1. We can leverage this redundancy to choose a ``standard form'' for each stack's winding number signs. We opt for the simple rule of choosing the representation which is largest according to the dictionary order relation on $\mathbb{Z}^6$. For example, from
\begin{equation*}
	({-3},{-1},0,{-1},1,{-2}) < ({-3},{-1},0,1,{-1},2) < (3,1,0,{-1},{-1},2) < (3,1,0,1,1,{-2})
\end{equation*}
the standard form would be $w^i=(3,1,0,1,1,{-2})$.

There are also simple operations which flip the signs of the winding numbers on rectangular tori which extend naturally to the case of tilted tori:
\begin{equation}
\begin{aligned}
	\operatorname{flip}_n:\quad(n^i,m^i) &\mapsto (-n^i,m^i+2b_in^i) & \quad\text{i.e.}\quad (n^i,\widehat{m}^i) &\mapsto (-n^i,\widehat{m}^i) \,,\\
	\operatorname{flip}_{\widehat{m}}:\quad (n^i,m^i) &\mapsto (n^i,-m^i-2b_in^i) & \text{i.e.}\quad (n^i,\widehat{m}^i) &\mapsto (n^i,-\widehat{m}^i) \,,\\
	\operatorname{flip}_{n\widehat{m}}:\quad(n^i,m^i) &\mapsto (-n^i,-m^i) & \text{i.e.}\quad (n^i,\widehat{m}^i) &\mapsto (-n^i,-\widehat{m}^i) \,.
\end{aligned}
\end{equation}
Notice that even for tilted tori these do not alter the coprimality conditions, since $\gcd(n,m+n)=\gcd(n,m)$. Each such flip results in an even number each of $\widehat{X}^I$ and $\widehat{Y}^I$ changing signs; as such this is not necessarily a symmetry of the solutions but rather a natural operation that will prove useful in the following section.

\begin{algorithm}[t]
	\caption{Genetic algorithm with elitism}\label{alg:GAelite}
	\begin{algorithmic}
		\State $\pi_0 \gets \textsc{RandomPopulation}(n_\text{pop})$
		\For{$g \gets 1$\text{ to }$g_\text{max}$}
			\State $\pi_g \gets \{\}$
			\State \textbf{append} $n_\text{elite}$ fittest individuals of $\pi_{g-1}$ to $\pi_g$
			\While{$\textsc{Size}(\pi_g)<n_\text{pop}$}
				\State $p_1 \gets \textsc{BinaryTournament}(\pi_{g-1})$
				\State $p_2 \gets \textsc{BinaryTournament}(\pi_{g-1})$
				\State $c \gets \textsc{CrossOver}(p_1,p_2)$
				\State $c \gets \textsc{Mutate}(c)$
				\State $c \gets \textsc{Adjust}(c)$ \Comment{e.g.\ ensure coprime winding numbers}
				\State \textbf{append} $c$ to $\pi_g$
			\EndWhile
		\EndFor
	\end{algorithmic}
\end{algorithm}


\section{Genetic algorithm}
\label{sec:GA}

In this section we describe the structure of the genetic algorithm used to search for solutions to the consistency conditions discussed in Sec.~\ref{sec:review}. We use a genetic algorithm with elitism, as outlined in Alg.~\ref{alg:GAelite}, in which a population of $n_\text{pop}$ individuals $\zeta_i$, each consisting of a chromosome $\chi(\zeta_i)$ of integer-valued genes and fitness $\mathcal{F}(\zeta_i)$, breeds to produce subsequent generations. The fittest $n_\text{elite}$ individuals survive to the next generation, and to fill out the rest of the population children are produced one-by-one from pairs of randomly selected parents. We use binary tournaments to select parents, where from two randomly selected individuals the fittest of the two is selected to breed. The child's chromosome is formed using one of a collection of cross-over methods whereby the parents' genes are spliced together. Finally, random mutations are applied to the child before they are added to the next generation. In the remainder of this section we will discuss in more detail how each of these steps is realized for the problem at hand.


\subsection{Environments}
The environment in which the genetic algorithm operates is determined by a choice of torus shapes, $b_i\in\{0,\frac{1}{2}\}$, and moduli $\widehat{U}_I$. In addition, the classification discussed in Sec.~\ref{sec:braneclassification} suggests possible reductions of the search space:
\begin{itemize}
	\item $\env=1$: All brane types are allowed.
	\item $\env=2$: Type $A'$, $B'$ and $C'$ branes are changed to types $A$, $B$, $C$ by a reassignment of winding number signs only (see App.~\ref{app:adjustments}).
	\item $\env=3$: Same as $\env=2$, but in addition all type $C$ and $C'$ branes are removed from the chromosome. Type $C$ branes are then added for the purposes of evaluating an individual's fitness.
\end{itemize}
We also place bounds on the number and size of stacks that can be described by an individual so that the search space is finite-dimensional. We place no bounds on the winding numbers.

As demonstrated in Sec.~\ref{sec:review}, the moduli $\widehat{U}_I$ can always be taken to be positive, coprime integers. These moduli are fixed for each run of the genetic algorithm so that all individuals in the population are attempting to solve the same (fixed) SUSY conditions. This means that in order to find solutions with different values of the moduli we must scan through various choices, running the genetic algorithm a number of times for each. There remains the question of which $\widehat{U}_I$ admit consistent solutions. One can make a rough argument that there is an upper bound for these integer moduli of the form $\widehat{U}_I<\mathcal{O}(2^{14})$, and certainly one would expect that there are far fewer models to be found when the moduli are large or do not share many factors. We leave a more detailed analysis of such bounds and their implications for the number of solutions for the future: for the time being we will restrict attention to $\widehat{U}_I\leq 5$.


\subsection{Chromosome structure \& cross-over}
\label{sec:xover}

\begin{figure}[t]
	\centering
	\newcommand{\s}{0.7}
	\newcommand{\hh}{1}
	\newcommand{\vv}{2.4}
	\begin{tikzpicture}
		\node[right] at (\hh + 3.4+4*\s+1.5,-0.5) {\phantom{\textbf{1su:}}}; 
		\node[left] at (-\hh,-0.5) {\textbf{1su:}};
		\foreach \r in {0,...,4} {
		\fill[red!60!black] (0,-0.25*\r) rectangle (0.2,-0.25*\r-0.2);
		\foreach \c in {1,...,6} {
			\fill[red] (0.25*\c,-0.25*\r) rectangle (0.25*\c+0.2,-0.25*\r-0.2);
		}
		}
		\foreach \r in {0,...,6} {
		\fill[cyan!60!black] (1.7+2*\s,-0.25*\r) rectangle (1.7+2*\s+0.2,-0.25*\r-0.2);
		\foreach \c in {1,...,6} {
			\fill[cyan] (1.7+2*\s+0.25*\c,-0.25*\r) rectangle (1.7+2*\s+0.25*\c+0.2,-0.25*\r-0.2);
		}
		}
		\draw[very thick] (-0.1,-0.475) -- (1.8,-0.475);
		\draw[very thick] (1.7+2*\s-0.1,-1.225) -- (3.4+2*\s+0.1,-1.225);
		\foreach \r in {0,1} {
		\fill[red!60!black] (3.4+4*\s,-0.25*\r) rectangle (3.4+4*\s+0.2,-0.25*\r-0.2);
		\foreach \c in {1,...,6} {
			\fill[red] (3.4+4*\s+0.25*\c,-0.25*\r) rectangle (3.4+4*\s+0.25*\c+0.2,-0.25*\r-0.2);
		}
		}
		\foreach \r in {2,3} {
		\fill[cyan!60!black] (3.4+4*\s,-0.25*\r) rectangle (3.4+4*\s+0.2,-0.25*\r-0.2);
		\foreach \c in {1,...,6} {
			\fill[cyan] (3.4+4*\s+0.25*\c,-0.25*\r) rectangle (3.4+4*\s+0.25*\c+0.2,-0.25*\r-0.2);
		}
		}
		\node at (1.7+\s,-0.5) {$\otimes$};
		\node at (3.4+3*\s,-0.5) {$=$};
		
		\node[left] at (-\hh,-\vv-0.5) {\textbf{1gc:}};
		\foreach \r in {0,...,4} {
		\fill[red!60!black] (0,-\vv-0.25*\r) rectangle (0.2,-\vv-0.25*\r-0.2);
		\foreach \c in {1,...,6} {
			\fill[red] (0.25*\c,-\vv-0.25*\r) rectangle (0.25*\c+0.2,-\vv-0.25*\r-0.2);
		}
		}
		\foreach \r in {0,...,6} {
		\fill[cyan!60!black] (1.7+2*\s,-\vv-0.25*\r) rectangle (1.7+2*\s+0.2,-\vv-0.25*\r-0.2);
		\foreach \c in {1,...,6} {
			\fill[cyan] (1.7+2*\s+0.25*\c,-\vv-0.25*\r) rectangle (1.7+2*\s+0.25*\c+0.2,-\vv-0.25*\r-0.2);
		}
		}
		\draw[very thick] (-0.1,-1*\vv-0.225) -- (1.475,-1*\vv-0.225) -- (1.475,-1*\vv+0.025) -- (1.8,-1*\vv+0.025);
		\draw[very thick] (1.7+2*\s-0.1,-1*\vv-0.725) -- (1.7+2*\s+1.475,-1*\vv-0.725) -- (1.7+2*\s+1.475,-1*\vv-0.475) -- (3.4+2*\s+0.1,-1*\vv-0.475);
		\fill[red!60!black] (3.4+4*\s,-\vv) rectangle (3.4+4*\s+0.2,-\vv-0.2);
		\foreach \c in {1,...,5} {
			\fill[red] (3.4+4*\s+0.25*\c,-1*\vv) rectangle (3.4+4*\s+0.25*\c+0.2,-1*\vv-0.2);
		}
		\foreach \r in {1,...,4} {
		\fill[cyan!60!black] (3.4+4*\s,-1*\vv-0.25*\r) rectangle (3.4+4*\s+0.2,-1*\vv-0.25*\r-0.2);
		\foreach \c in {1,...,6} {
			\fill[cyan] (3.4+4*\s+0.25*\c,-1*\vv-0.25*\r) rectangle (3.4+4*\s+0.25*\c+0.2,-1*\vv-0.25*\r-0.2);
		}
		}
		\fill[cyan] (3.4+4*\s+0.25*6,-1*\vv) rectangle (3.4+4*\s+0.25*6+0.2,-1*\vv-0.2);
		\node at (1.7+\s,-\vv-0.5) {$\otimes$};
		\node at (3.4+3*\s,-\vv-0.5) {$=$};

		\node[left] at (-\hh,-2*\vv-0.5) {\textbf{2sc:}};
		\foreach \r in {0,...,4} {
		\fill[red!60!black] (0,-2*\vv-0.25*\r) rectangle (0.2,-2*\vv-0.25*\r-0.2);
		\foreach \c in {1,...,6} {
			\fill[red] (0.25*\c,-2*\vv-0.25*\r) rectangle (0.25*\c+0.2,-2*\vv-0.25*\r-0.2);
		}
		}
		\foreach \r in {0,...,6} {
		\fill[cyan!60!black] (1.7+2*\s,-2*\vv-0.25*\r) rectangle (1.7+2*\s+0.2,-2*\vv-0.25*\r-0.2);
		\foreach \c in {1,...,6} {
			\fill[cyan] (1.7+2*\s+0.25*\c,-2*\vv-0.25*\r) rectangle (1.7+2*\s+0.25*\c+0.2,-2*\vv-0.25*\r-0.2);
		}
		}
		\draw[very thick] (-0.1,-2*\vv-0.475) -- (1.8,-2*\vv-0.475);
		\draw[very thick] (-0.1,-2*\vv-0.975) -- (1.8,-2*\vv-0.975);
		\draw[very thick] (1.7+2*\s-0.1,-2*\vv-0.225) -- (3.4+2*\s+0.1,-2*\vv-0.225);
		\draw[very thick] (1.7+2*\s-0.1,-2*\vv-0.725) -- (3.4+2*\s+0.1,-2*\vv-0.725);
		\foreach \r in {0,1,4} {
		\fill[red!60!black] (3.4+4*\s,-2*\vv-0.25*\r) rectangle (3.4+4*\s+0.2,-2*\vv-0.25*\r-0.2);
		\foreach \c in {1,...,6} {
			\fill[red] (3.4+4*\s+0.25*\c,-2*\vv-0.25*\r) rectangle (3.4+4*\s+0.25*\c+0.2,-2*\vv-0.25*\r-0.2);
		}
		}
		\foreach \r in {2,3} {
		\fill[cyan!60!black] (3.4+4*\s,-2*\vv-0.25*\r) rectangle (3.4+4*\s+0.2,-2*\vv-0.25*\r-0.2);
		\foreach \c in {1,...,6} {
			\fill[cyan] (3.4+4*\s+0.25*\c,-2*\vv-0.25*\r) rectangle (3.4+4*\s+0.25*\c+0.2,-2*\vv-0.25*\r-0.2);
		}
		}
		\node at (1.7+\s,-2*\vv-0.5) {$\otimes$};
		\node at (3.4+3*\s,-2*\vv-0.5) {$=$};
		
		\node[left] at (-\hh,-3*\vv-0.5) {\textbf{2gu:}};
		\foreach \r in {0,...,4} {
		\fill[red!60!black] (0,-3*\vv-0.25*\r) rectangle (0.2,-3*\vv-0.25*\r-0.2);
		\foreach \c in {1,...,6} {
			\fill[red] (0.25*\c,-3*\vv-0.25*\r) rectangle (0.25*\c+0.2,-3*\vv-0.25*\r-0.2);
		}
		}
		\foreach \r in {0,...,6} {
		\fill[cyan!60!black] (1.7+2*\s,-3*\vv-0.25*\r) rectangle (1.7+2*\s+0.2,-3*\vv-0.25*\r-0.2);
		\foreach \c in {1,...,6} {
			\fill[cyan] (1.7+2*\s+0.25*\c,-3*\vv-0.25*\r) rectangle (1.7+2*\s+0.25*\c+0.2,-3*\vv-0.25*\r-0.2);
		}
		}
		\draw[very thick] (-0.1,-3*\vv-0.225) -- (0.475,-3*\vv-0.225) -- (0.475,-3*\vv+0.025) -- (1.8,-3*\vv+0.025);
		\draw[very thick] (-0.1,-3*\vv-0.725) -- (0.975,-3*\vv-0.725) -- (0.975,-3*\vv-0.475) -- (1.8,-3*\vv-0.475);
		\draw[very thick] (1.7+2*\s-0.1,-3*\vv-0.725) -- (1.7+2*\s+0.475,-3*\vv-0.725) -- (1.7+2*\s+0.475,-3*\vv-0.475) -- (3.4+2*\s+0.1,-3*\vv-0.475);
		\draw[very thick] (1.7+2*\s-0.1,-3*\vv-1.475) -- (1.7+2*\s+0.975,-3*\vv-1.475) -- (1.7+2*\s+0.975,-3*\vv-1.225) -- (3.4+2*\s+0.1,-3*\vv-1.225);
		\fill[red!60!black] (3.4+4*\s,-3*\vv) rectangle (3.4+4*\s+0.2,-3*\vv-0.2);
		\fill[red] (3.4+4*\s+0.25,-3*\vv) rectangle (3.4+4*\s+0.25+0.2,-3*\vv-0.2);
		\foreach \c in {2,...,6} {
			\fill[cyan] (3.4+4*\s+0.25*\c,-3*\vv) rectangle (3.4+4*\s+0.25*\c+0.2,-3*\vv-0.2);
		}
		\fill[cyan!60!black] (3.4+4*\s,-3*\vv-0.75) rectangle (3.4+4*\s+0.2,-3*\vv-0.75-0.2);
		\foreach \c in {1,2,3} {
			\fill[cyan] (3.4+4*\s+0.25*\c,-3*\vv-0.75) rectangle (3.4+4*\s+0.25*\c+0.2,-3*\vv-0.75-0.2);
		}
		\foreach \r in {1,2} {
		\fill[cyan!60!black] (3.4+4*\s,-3*\vv-0.25*\r) rectangle (3.4+4*\s+0.2,-3*\vv-0.25*\r-0.2);
		\foreach \c in {1,...,6} {
			\fill[cyan] (3.4+4*\s+0.25*\c,-3*\vv-0.25*\r) rectangle (3.4+4*\s+0.25*\c+0.2,-3*\vv-0.25*\r-0.2);
		}
		}
		\foreach \c in {4,5,6} {
			\fill[red] (3.4+4*\s+0.25*\c,-3*\vv-0.75) rectangle (3.4+4*\s+0.25*\c+0.2,-3*\vv-0.75-0.2);
		}
		\foreach \r in {4,5} {
		\fill[red!60!black] (3.4+4*\s,-3*\vv-0.25*\r) rectangle (3.4+4*\s+0.2,-3*\vv-0.25*\r-0.2);
		\foreach \c in {1,...,6} {
			\fill[red] (3.4+4*\s+0.25*\c,-3*\vv-0.25*\r) rectangle (3.4+4*\s+0.25*\c+0.2,-3*\vv-0.25*\r-0.2);
		}
		}
		\node at (1.7+\s,-3*\vv-0.5) {$\otimes$};
		\node at (3.4+3*\s,-3*\vv-0.5) {$=$};
	\end{tikzpicture}
	\caption{Schematic representation of examples for four of the eight cross-over method described in Sec.~\ref{sec:xover}. Each row in a chromosome contains the stack size (darker shade) and six winding numbers (ligher shade) of a brane stack. The chromosomes of parent 1 (red) and parent 2 (cyan) are split at either one or two randomly chosen point(s) and combined.}
	\label{fig:xovermethods}
\end{figure}
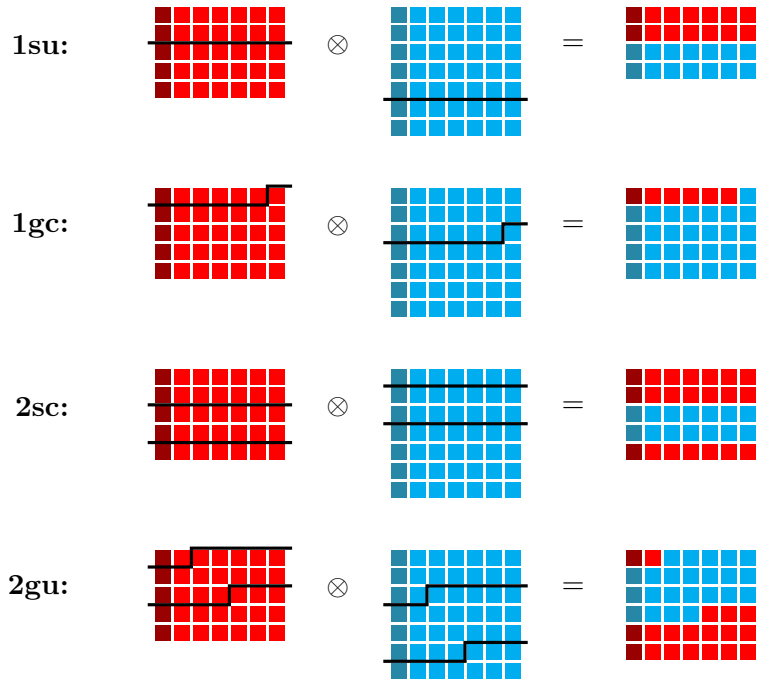

An individual $\zeta$ consists of a number of stacks $k\in\{k_\text{min},\ldots,k_\text{max}\}$ with each D6-brane stack consisting of a stack size $N_a\in\{1,\ldots,N_\text{max}\}$ and six integer winding numbers $n_a^i,m_a^i$ which are pairwise coprime. The chromosome can be organized into a $k\times 7$ matrix as follows:
\begin{equation}
	\chi(\zeta) = \left[\begin{array}{c|cccccc}
		N_1 & n_1^1 & m_1^1 & n_1^2 & m_1^2 & n_1^3 & m_1^3\\[5pt]
		N_2 & n_2^1 & m_2^1 & n_2^2 & m_2^2 & n_2^3 & m_2^3\\[5pt]
		\vdots & \vdots & \vdots & \vdots & \vdots & \vdots & \vdots\\[5pt]
		N_k & n_k^1 & m_k^1 & n_k^2 & m_k^2 & n_k^3 & m_k^3
	\end{array}\right]_{k\times 7}
\end{equation}
We employ eight cross-over methods to combine the chromosomes of the two selected parents. These methods differ in three regards; (i) the number of cross-over points (either one or two), (ii) where cross-over points are allowed to be (only between stacks or between any gene), and (iii) whether the number of stacks in the child is constrained to match that of one of the parents. These three binary choices give us eight methods, which we label by (i) \textbf{1}/\textbf{2}, (ii) \textbf{s}/\textbf{g} (stack/gene), and (iii) \textbf{c}/\textbf{u} (constrained/unconstrained). Four of these eight methods are exemplified in Fig.~\ref{fig:xovermethods}.


\subsection{Mutations}
\label{sec:mutations}
After a child's chromosome is created via cross-over, a number of mutations $\mu_i$ are potentially applied to it. These mutations provide small changes which help both to keep the population diverse and to make gradual improvements to the fitness. The following list of mutations are used:
\begin{itemize}
	\item \textbf{N:} Change a stack's size, $\mu_N(N) = N\pm 1$.
	\item \textbf{w:} Change a winding number, $\mu_w(w) = w\pm\{1,2\}$.
	\item $\boldsymbol{\pm}$\textbf{:} Randomize a stack's winding numbers' signs. See App.~\ref{app:mutations} for details.
	\item $\mathbf{\boldsymbol{\mathfrak{S}}_4}$\textbf{:} Permute a stack's $\widehat{X}^I,\widehat{Y}^I$. See App.~\ref{app:mutations} for details.
\end{itemize}
Because the number of stacks $k$ is variable, to each mutation $\mu_i$ is associated a rate $r_i$ which parametrizes the expected number of times this mutation is applied. That is, $\mu_w$ is applied to each winding number with probability $r_w/6k$, while $\mu_N$, $\mu_\pm$ and $\mu_{\mathfrak{S}_4}$ are applied to each stack with probabilities $r_N/k$, $r_\pm/k$ and $r_{\mathfrak{S}_4}/k$, respectively. As a heuristic, genetic algorithms perform best when the number of applied mutations is $\mathcal{O}(1)$, so we should anticipate the best results when $\sum_ir_i = \mathcal{O}(1)$.


\subsection{Adjustments}
Cross-over and application of mutations may result in a chromosome which is not ``valid'', i.e.\ falls outside the scope of the chosen environment. As a final step before being added to the population, the following checks are made in turn:
\begin{itemize}
	\item Non-empty stacks: each stack should have $N_a\geq 1$.
	\item Coprimality: each pair of winding numbers $(n,m)$ should satisfy $\gcd(n,m)=1$.
	\item Standardization: bring each stack's winding numbers to the ``standard form'' discussed in Sec.~\ref{sec:symsandoperations}.
	\item Allowed types: for $\env=2$ and $\env=3$ stacks of type $A'$, $B'$, $C'$ and $C$ are potentially changed or removed.
	\item No duplicates: stacks with identical winding numbers are combined.
\end{itemize}
See App.~\ref{app:adjustments} for more details on how such changes are implemented.


\subsection{Fitness function} 
The fitness function measures how close an individual is to having desirable properties. This includes tadpole cancellation (T), vanishing of K-theory charge (K) and preservation of supersymmetry (S), as well as MSSM-like properties (M). We take the fitness function to have the form
\begin{equation}
	\mathcal{F}(\zeta) = \mathcal{W}_\text{T}\mathcal{F}_\text{T}(\zeta) + \mathcal{W}_\text{K}\mathcal{F}_\text{K}(\zeta) + \mathcal{W}_\text{S}\mathcal{F}_\text{S}(\zeta) + \mathcal{W}_\text{M}\mathcal{F}_\text{M}(\zeta) \,,
\end{equation}
where each $\mathcal{F}_i$ lies in the interval $[0,1]$ with $\mathcal{F}_i=0$ being optimal. We can choose the non-negative weights $\mathcal{W}_i$ to add to one without loss of generality. The individual terms are taken to be of the following forms,
\begin{align}
    \mathcal{F}_\text{T}(\zeta) &= h\Big(\tfrac{\langle T^I\rangle}{\Delta_\text{T}}\Big) \,, & T^I &= \Big|\sum_aN_a\widehat{X}_a^I - 8\Big| \,,\\
    \mathcal{F}_\text{K}(\zeta) &= \langle K^I\rangle \,, & K^I &= \Big(\sum_a N_a\widehat{Y}_a^I\Big) \mod{2} \,,\\
    \mathcal{F}_\text{S}(\zeta) &= h\Big(\tfrac{\langle S_a\rangle}{\Delta_\text{S}}\Big) \,, & S_a &= \left|\min\left\{\frac{\sum_I\widehat{X}_a^I\widehat{U}_I}{\sum_I\widehat{U}_I},0\right\}\right| + \left|\frac{\sum_I \widehat{Y}_a^I/\widehat{U}_I}{\sum_I 1/\widehat{U}_I}\right| \,,\\
    \mathcal{F}_\text{M}(\zeta) &= G(\zeta) \,,
\end{align}
where the brackets $\langle\,\cdot\,\rangle$ indicate an average over either $I$ or $a$, as appropriate, and $h(z)=\frac{z}{1+z}$ tames the fitness for large tadpole or large $S_a$; the ``scales'' $\Delta_\text{T}$ and $\Delta_\text{S}$ control the points at which the nonlinearity sets in. Note that it is important to average over $S_a$ rather than take a sum because the number of stacks each individual has is variable; a sum would artificially favor individuals with fewer stacks. The function $G$ measures how close the individual is to having MSSM-like characteristics. We take perhaps the simplest choice and choose $G$ to measure how far the gauge group is from containing a $\U(3)\times\U(2)\times\U(1)\times\U(1)$ factor:
\begin{equation}
    G(\zeta) = \frac{1}{4}\Big(\max\!\big\{1-\#\U(3),0\big\} + \max\!\big\{1-\#\U(2),0\big\} + \max\!\big\{2-\#\U(1),0\big\}\Big) \,.
\end{equation}
The function $G$ can also be tailored to weight additional features such as number of families, number of exotics, etc.

We will refer to individuals or models collectively in terms of which of the three consistency conditions they satisfy. For example, KS models satisfy the K-theory and SUSY conditions but have one or more uncancelled tadpoles, while TKS models are fully consistent.


\subsection{Summary}
For each instance of the genetic algorithm described above there are several parameters which need to be chosen. These fall into two categories, which we will refer to as ``environment parameters'' and ``hyper-parameters'':
\begin{itemize}
    \item Environment parameters: $b_i$, $\widehat{U}_I$, $k_\text{min}$, $k_\text{max}$, $N_\text{max}$, $\env$
    \item Hyper-parameters: $n_\text{pop}$, $n_\text{elite}$, $g_\text{max}$, $p_i$, $r_i$, $\mathcal{W}_i$, $\Delta_\text{T,S}$
\end{itemize}
Environment parameters are chosen by hand and to search the landscape one should conduct runs for various choices of $b_i$ and $\widehat{U}_I$. The hyper-parameters should be chosen to maximize the genetic algorithm's efficiency in finding the desired solutions. In App.~\ref{app:optimization} we discuss the process by which these hyper-parameters are chosen to be
\begin{equation}\label{eq:optparams}
\begin{aligned}
    (n_\text{pop}, n_\text{elite}, g_\text{max}) &= (250,\; 25,\; 100)\\
    (p_\textbf{1sc}, p_\textbf{1su}, p_\textbf{1gc}, p_\textbf{1gu}) &= (0.19,\; 0.05,\; 0.29,\; 0.07)\\
    (p_\textbf{2sc}, p_\textbf{2su}, p_\textbf{2gc}, p_\textbf{2gu}) &= (0.13,\; 0.03,\; 0.19,\; 0.05)\\
    (r_N,r_w,r_\pm,r_{\mathfrak{S}_4}) &= (0.07,\; 0.28,\;0.17,\;0.19)\\
    (\mathcal{W}_\text{T},\mathcal{W}_\text{K},\mathcal{W}_\text{S}) &= (1-\mathcal{W}_\text{M})\times (0.28,\; 0.05,\; 0.67)\\
    (\Delta_\text{T},\Delta_\text{S}) &= (21,\;0.90)
\end{aligned}
\end{equation}


\section{Results}
\label{sec:results}

\begin{figure}[t]
    \centering
    \includegraphics[width=0.8\textwidth]{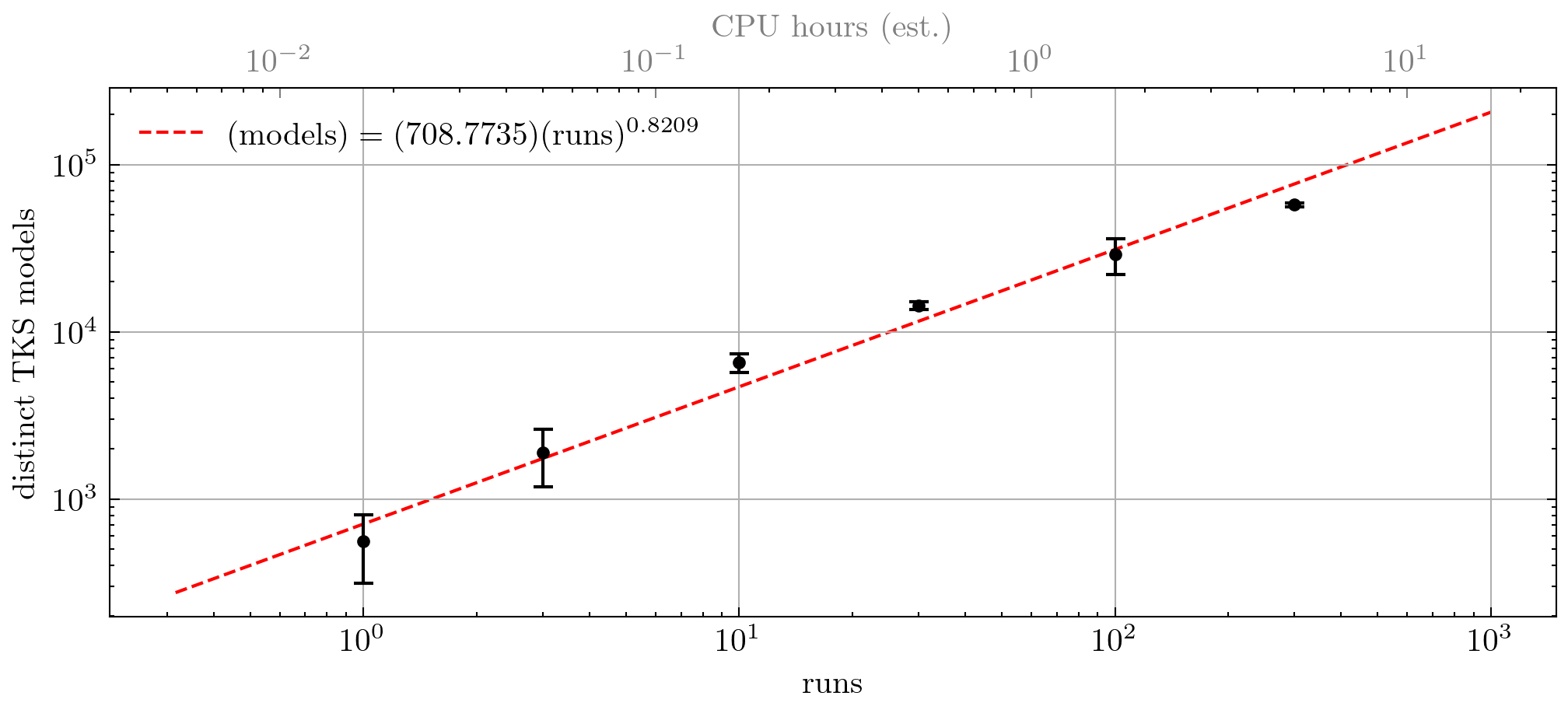}
    \caption{Number of distinct models models found for $b_i=0$, $\widehat{U}_I=1$, $\env=3$ as a function of the number of runs (along with an estimate of the corresponding number of CPU hours). The number of models grows sublinearly due to a growing chance of finding duplicate models.}
    \label{fig:modelsfound}
\end{figure}

In this section we present the three ensembles of models found during three searches using the hyper-parameters of Eqn.~\eqref{eq:optparams}. The remaining parameters are chosen to be
\begin{equation*}\arraycolsep=10pt
    \begin{array}{l|ccccc}
        & \env & \mathcal{W}_\text{M} & k_\text{min} & k_\text{max} & N_\text{max}\\[3pt] \hline
        \text{Search I} & 1 & 0 & 3 & 10 & 10\\
        \text{Search II} & 3 & 0 & 3 & 10 & 10\\
        \text{Search III} & 3 & \phantom{.05}0.05 & 4 & 10 & 10\\
    \end{array}
\end{equation*}
For each of search I, II and III we scan over all 1052 inequivalent combinations of $b_i$ and $\widehat{U}_I\leq 5$,\footnote{Take $b_i$ weakly increasing and if $b_i=b_j$ with $i<j$, then $\widehat{U}_i\leq \widehat{U}_j$.} and perform $1000$ runs for each. The number of distinct models (e.g.\ accounting for permutations of stacks in the chromosomes) found for each search are given in the following table:
\begin{equation*}
    \begin{tabular}{l|cccc}
         & KS models & TKS models & TKS models w/ $G(\zeta)=0$\\ \hline
        Search I & 82,799,762 & \phantom{3,}110,783 & \phantom{726,}698\\
        Search II & -- & 3,172,276 & 107,101\\
        Search III & -- & 3,069,750 & 726,837
    \end{tabular}
\end{equation*}

\begin{figure}[t]
    \centering
    \includegraphics[width=\textwidth]{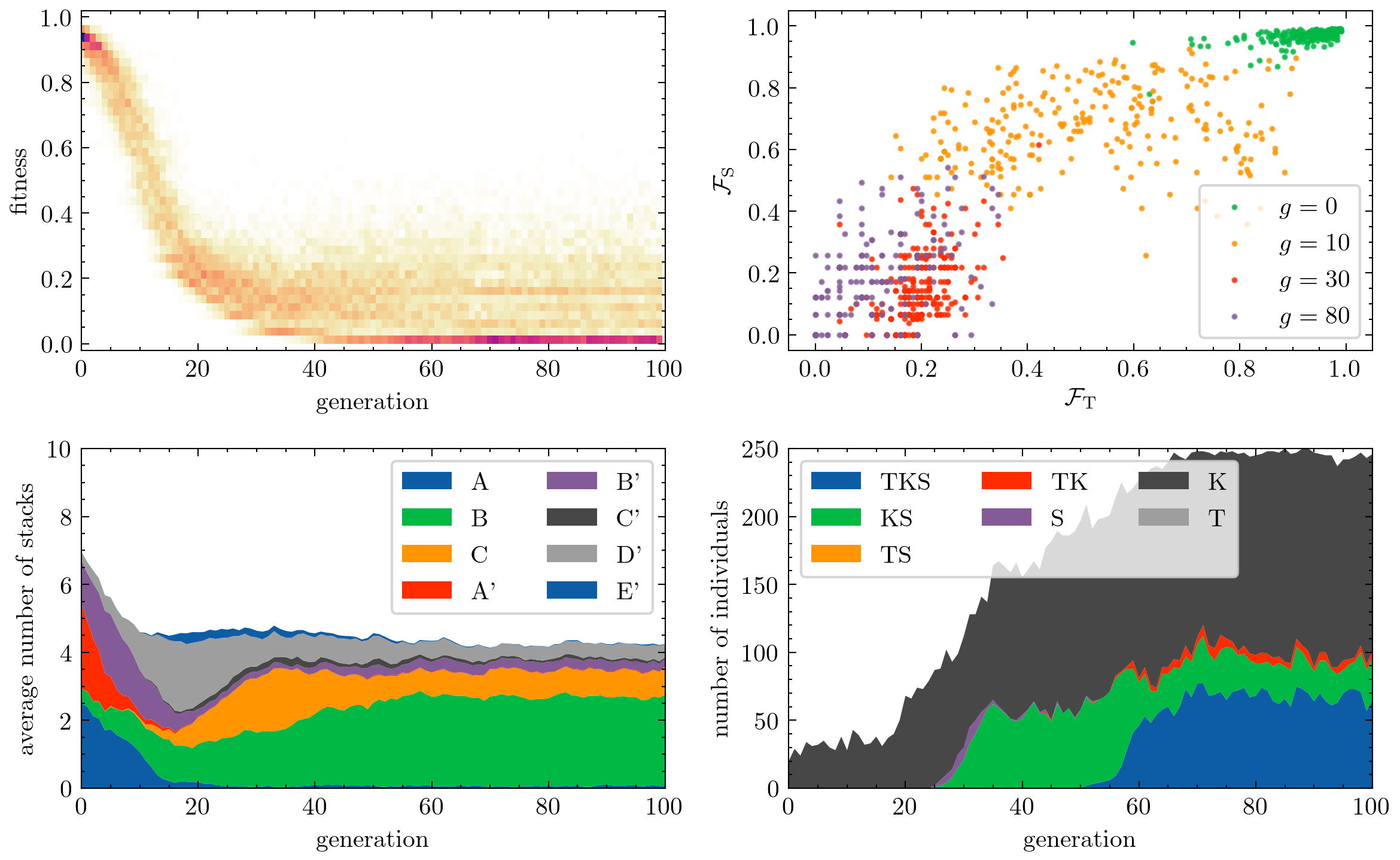}
    \caption{Evolution of a population during a typical run of the genetic algorithm: distribution of fitnesses (top-left), distribution tadpole and supersymmetry fitnesses (top-right), proportion of branes of each type (bottom-left), and number of models in the population satisfying one or more of the three consistency conditions (bottom-right).}
    \label{fig:singlerun}
\end{figure}

The structure of genetic algorithms is well-suited for parallelization. Crossover, mutation and evaluation of fitness may be run in parallel for each generation, and also many instances of the genetic algorithm may be run in parallel with their own populations evolving independently. A single run takes $\mathcal{O}(1)$ CPU minutes. Fig.~\ref{fig:modelsfound} displays the number of distinct models found for search II with $b_i=0$ and $\widehat{U}_I=1$ as the number of runs is increased.


\subsection{A typical run \& learning strategy}
\label{sec:learning}

Let us first discuss a typical run of the genetic algorithm. By looking at how the characteristics of the population change over the generations, one may learn about the strategy employed by the genetic algorithm when the parameters have been chosen optimally. In Fig.~\ref{fig:singlerun} are presented the results of a single run of search I. The first SUSY-preserving individual appears at generation 25 and the first fully consistent individual appears at generation 51. After 100 generations a total of 20 distinct TKS solutions have been found.

We can understand the learning strategy in a straightforward way by considering the evolution of the individual tadpole and SUSY terms in the fitness function and appearance of type $C$ branes. There is a ``burn-in'' period where the fitnesses dive from their initial values of $\mathcal{F}\approx 1$ from random initialization. The SUSY conditions apply to each brane stack individually, so there is always an incentive to have all stacks in the population satisfy these conditions, even during these early generations. Type $C$ branes appear and become commonplace in the population around the time the first SUSY-preserving individual is found. The tadpole conditions are subsequently solved as type $C$ branes are advantageously spread through the population via cross-over.


\subsection{Landscape statistics}

\afterpage{
\clearpage
\begin{figure}[t]
    \centering
    \includegraphics[width=\textwidth]{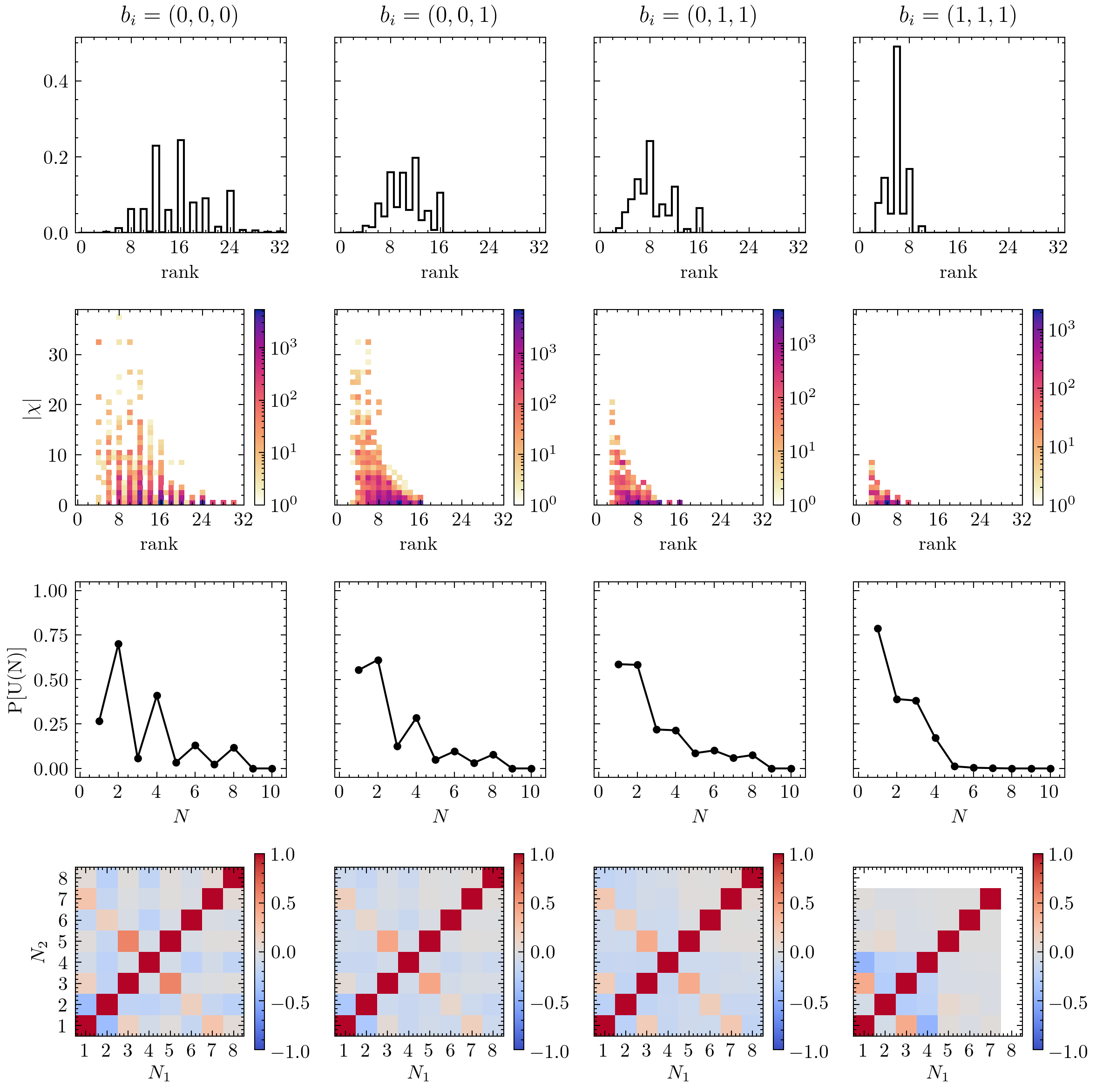}
    \caption{Statistics for search I. Each column corresponds to a different number of tilted tori. (Top row) Distributions of ranks. (Second row) Distributions of mean chirality vs.\ rank. (Third row) Proportion of models containing $\U(N)$ factors in their gauge group. (Bottom row) Correlations between the number of $\U(N_1)$ and $\U(N_2)$ factors.}
    \label{fig:searchI_TKSstats}
\end{figure}
\clearpage
}

\afterpage{
\clearpage
\begin{figure}[t]
    \centering
    \includegraphics[width=\textwidth]{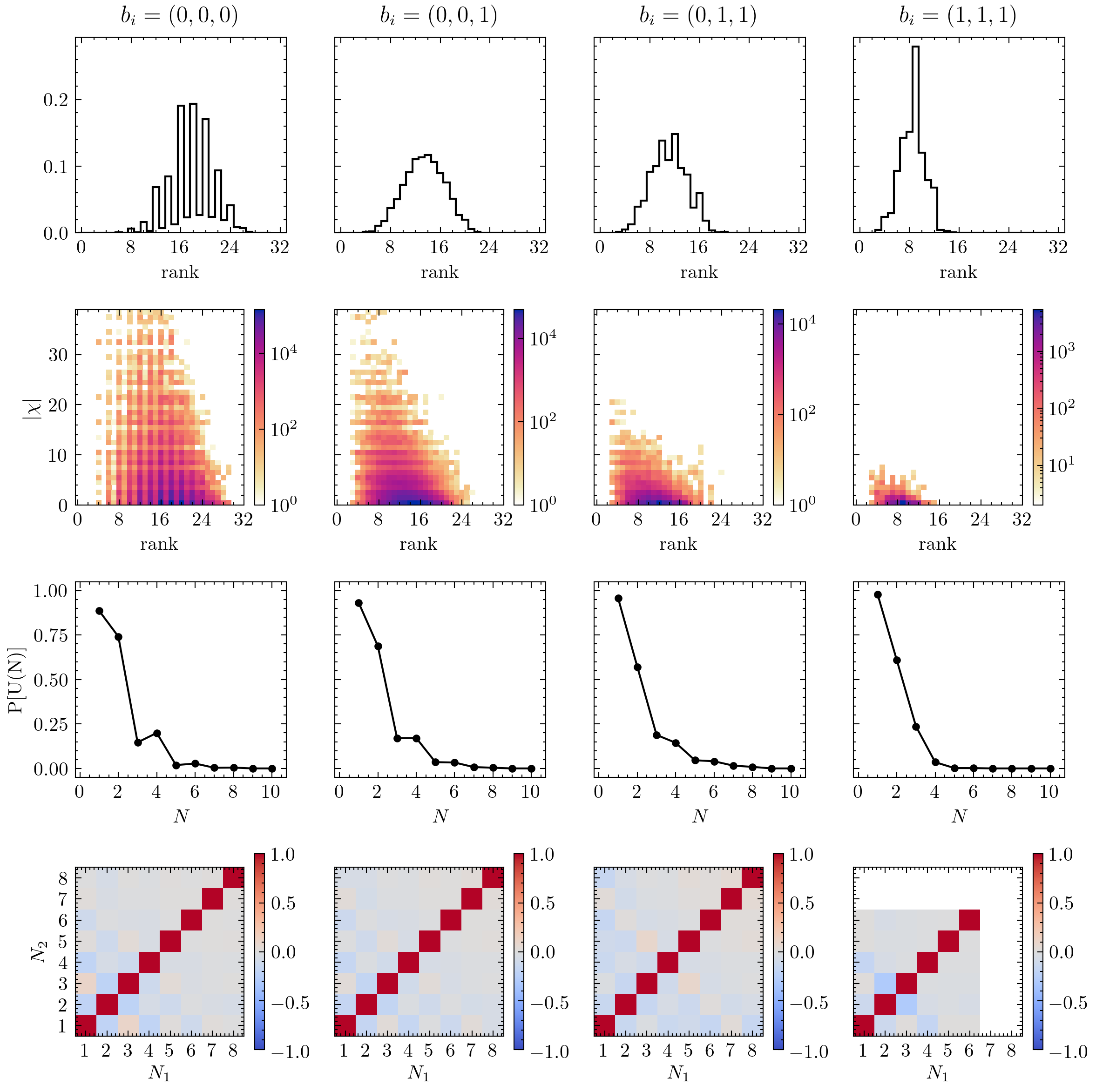}
    \caption{Statistics for search II. Each column corresponds to a different number of tilted tori. (Top row) Distributions of ranks. (Second row) Distributions of mean chirality vs.\ rank. (Third row) Proportion of models containing $\U(N)$ factors in their gauge group. (Bottom row) Correlations between the number of $\U(N_1)$ and $\U(N_2)$ factors.}
    \label{fig:searchII_TKSstats}
\end{figure}
\clearpage
}

With both search I and search II the genetic algorithm does not weight MSSM-like properties in the fitness function and only aims to satisfy the three consistency conditions. For the ensemble of representative models they find we may consider various statistics and their correlations. Figs.~\ref{fig:searchI_TKSstats} and \ref{fig:searchII_TKSstats} show the results for search I and search II, with the models corresponding to different numbers of tilted tori shown separately. The distribution of gauge group ranks has a strong suppression of odd rank when all tori are rectangular (as was observed in, e.g.,~\cite{Gmeiner:2005vz}). In addition, the distribution of ranks shifts to smaller and smaller values as the number of tilted tori increases. For tilted tori the tadpole conditions become more stringent and generally lead to smaller stack sizes, $N_a$, and thus smaller rank. There are only small correlations between the numbers of different $\U(N)$ factors, appearing in a distinctly checkerboard pattern due to the explicit eights in the tadpole conditions. There are no qualitative differences between the ensembles found during search I and search II.

Given the nature of the genetic algorithm that we use, we have access not only to fully consistent models but also those that are partially consistent which appear through the generations. This allows us to consider how the statistics change as the consistency conditions are satisfied one-by-one. In light of the learning strategy observed in Sec.~\ref{sec:learning}, the dependence of statistics for KS models on the margin by which the tadpole cancellation is violated is of particular interest. Fig.~\ref{fig:searchI_KSstats} shows how both the distribution of ranks and $\U(N)$ factors depends on $\langle T^I\rangle$ for a subset of KS models found during search I. The only qualitative change for $\langle T^I\rangle>0$ seems to be a wider spread in gauge group ranks, which is easily understood since the values of $N_a$ are less restricted by the tadpole conditions.


\subsection{MSSM-like statistics}
While the genetic algorithm has been successful in sampling the landscape of consistent models, the parameters for search III allow it to focus on consistent models which are not necessarily representative of the landscape as a whole. With $\mathcal{W}_M>0$ it will preferentially search for models which have a $\U(3)\times\U(2)\times\U(1)\times\U(1)$ gauge group factor \emph{in addition to} satisfying the three consistency conditions. Indeed, a large proportion TKS models now contain the desired gauge group factor.

Having found a large number of fully consistent models with a $\U(3)\times\U(2)\times\U(1)\times\U(1)$ gauge group factor, we may now ask about the more fine-grained phenomenological characteristics discussed in Sec.~\ref{sec:pheno}. For those models with more than one way to identify the four-stack MSSM gauge group we select the labeling which minimizes the variance in the number of quark families and lepton families (i.e.\ where $n_Q,n_u,n_d$ are most similar and $n_L,n_e$ are most similar).

Ratios of coupling constants are determined only by topological data, and their distribution is shown in Fig.~\ref{fig:couplings}. There is a mild positive correlation between $\alpha_Y/\alpha_s$ and $\alpha_w/\alpha_s$. For those models which have both $n_Q=n_u=n_d$ and $n_L=n_e$, the distribution of the numbers of families is shown in Fig.~\ref{fig:quarklepton}. There we see that for untilted tori ($b_i=0$) there are only ever an even number of families. With tilted tori odd numbers of families are possible and while we do find models with three families of quarks and three families of leptons, we have not found any with three of each simultaneously. It would be interesting to include these additional characteristics in the fitness function and further skew the sampling towards these phenomenologically exciting models: we leave such investigation for the future.

\begin{figure}[t]
    \centering
    \includegraphics[width=0.8\textwidth]{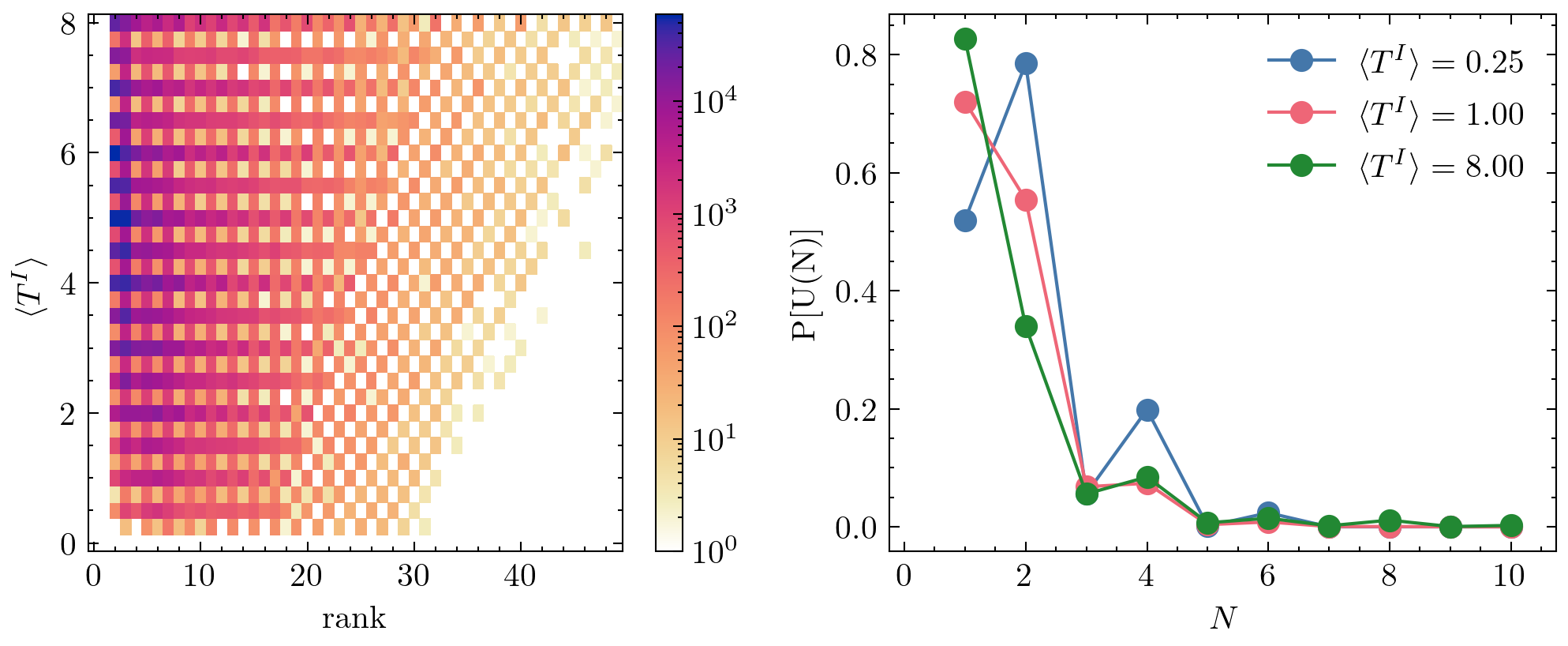}
    \caption{Dependence of gauge group rank and presence of $\U(N)$ gauge factors on degree of tadpole cancellation violation for a subset of 2,000,000 KS models found during search I.}
    \label{fig:searchI_KSstats}
\end{figure}

\begin{figure}[t]
    \centering
    \includegraphics[width=0.55\textwidth]{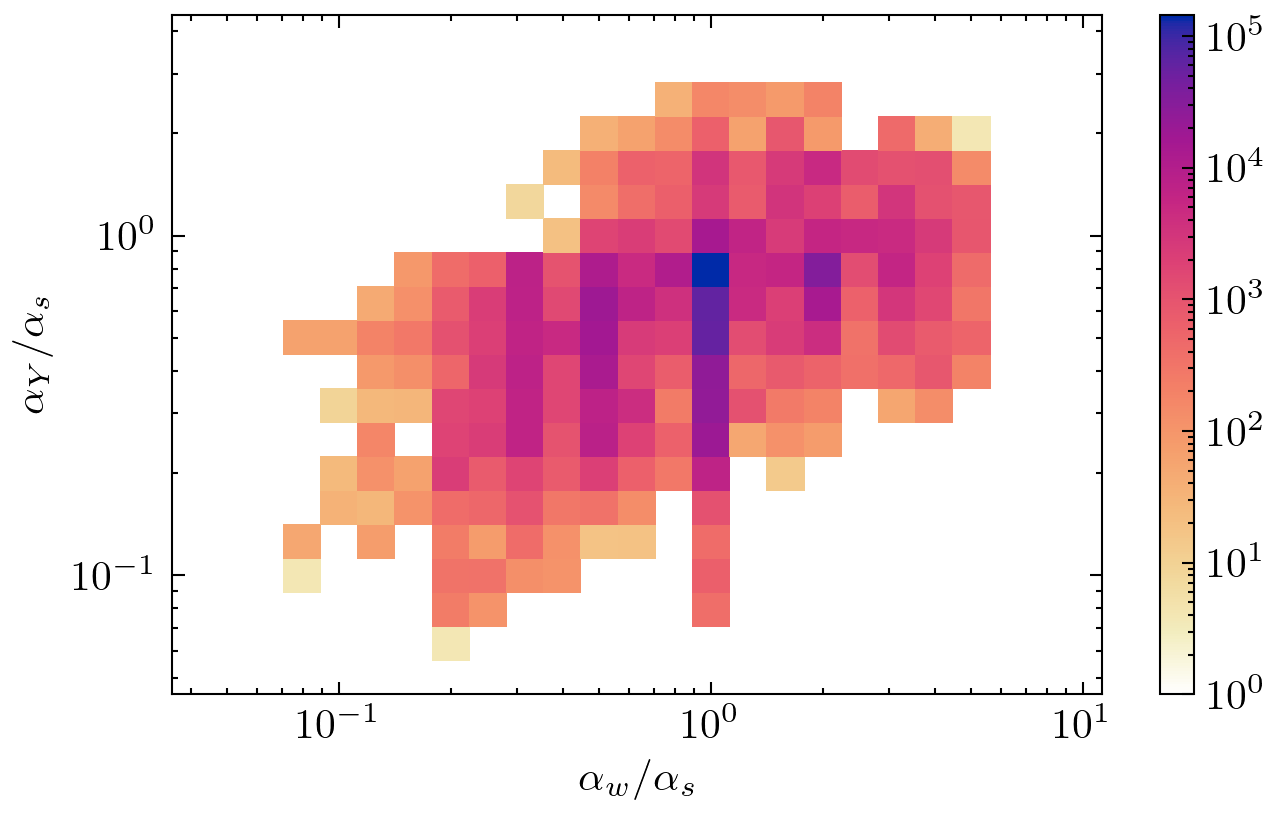}
    \caption{Distribution of coupling constant ratios for TKS models with MSSM gauge group factor found during search III.}
    \label{fig:couplings}
\end{figure}

\begin{figure}[t]
    \centering
    \includegraphics[width=0.9\textwidth]{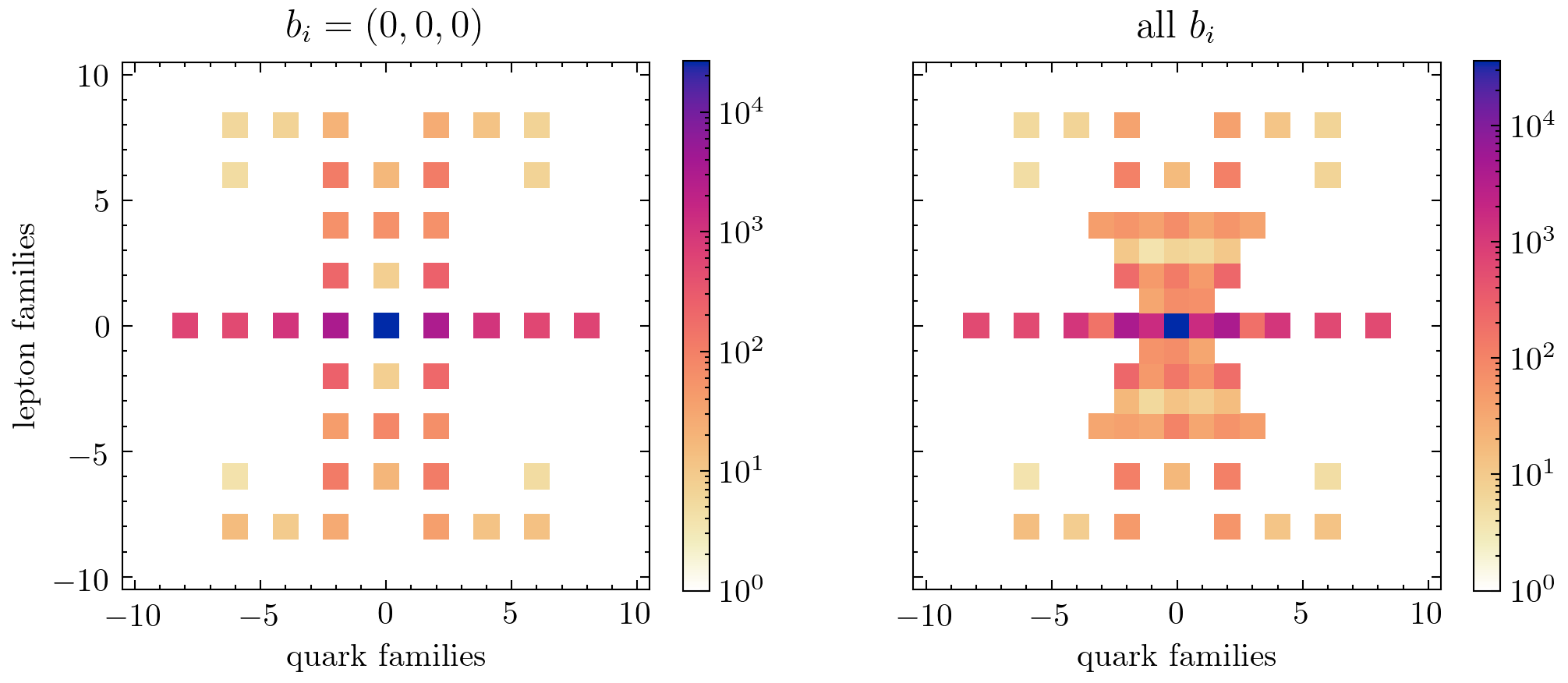}
    \caption{Distribution of number of families of quarks and leptons for TKS models with MSSM gauge group and both $n_Q=n_u=n_d$ and $n_L=n_e$. For all tori rectangular there are only ever an odd numbers of families.}
    \label{fig:quarklepton}
\end{figure}


\section{Discussion}
\label{sec:discussion}

In this paper we have demonstrated the utility of using genetic algorithms to explore the landscape of intersecting brane models for a particular type IIA toroidal orbifold compactification. Genetic algorithms lend themselves well to searching for the integer-valued models that live in the high-dimensional space describing intersecting branes. The purpose of conducting our genetic-based optimization in this simple orientifold is to learn {\it how} the algorithm learns to find optimal solutions. The methodology we developed here is directly transportable to other intersecting brane models. Moreover, Lemma \ref{positive-coprime-lemma} we proved in Section \ref{sec:review} applies to any search algorithm and can overcome the bottleneck encountered in previous computer searches for consistent branes.

We have conducted an initial foray into the landscape of consistent models by conducting three searches of slightly different nature. We were able to peer inside the genetic algorithm and understand the learning strategy used to solve the tadpole and SUSY conditions one-by-one. Large ensembles of consistent models provide a statistical way to understand the landscape, and we discussed several easily computable statistics which give a sense of the typical models the genetic algorithm is likely to find. It is an interesting question to determine how representative the sample of models found by the genetic algorithm is amongst all possible models. In particular, how does its searching compare with other machine learning techniques? We leave comparisons such as this (along the lines of~\cite{Cole:2021nnt,Abel:2021ddu}) for future investigations. Furthermore, our genetic algorithm can efficiently generate a sizable sample of intersecting brane models with the MSSM gauge factor. This allows us to explore the statistics of supersymmetry breaking as well the fine-tuning problem of electroweak symmetry breaking along similar veins as \cite{Susskind:2004uv,Douglas:2004qg,Dine:2004is,
Marchesano:2004yn}. We hope to report our findings in future publications.


\acknowledgments
We thank Alex Cole, Sven Krippendorf and Andreas Schachner for discussions.
The computations in this research were performed using the compute resources and assistance of the UW-Madison Center For High Throughput Computing (CHTC) in the Department of Computer Sciences. The CHTC is supported by UW-Madison, the Advanced Computing Initiative, the Wisconsin Alumni Research Foundation, the Wisconsin Institutes for Discovery, and the National Science Foundation, and is an active member of the OSG Consortium, which is supported by the National Science Foundation and the U.S.\ Department of Energy's Office of Science. The work of GL and GS is supported in part by the DOE grant DE-SC0017647 and the Kellett Award of the University of Wisconsin.


\appendix

\section{Genetic algorithm details}
\label{app:GA}

\subsection{Initialization}
\label{app:initialization}

The seed population consists of $n_\text{pop}$ random individuals. The number of stacks for each individual is chosen uniformly from $\{k_\text{min},\ldots,k_\text{max}\}$, and the genes in each stack are chosen according to the following distributions:\footnote{If $Y\sim\operatorname{Pois}(\mu_1)$ and $Z\sim\operatorname{Pois}(\mu_2)$ are Poisson random variables, then $X\equiv(Y-Z)\sim\operatorname{Skellam}(\mu_1,\mu_2)$ is a Skellam random variable. $X$ has mean $\mu_1-\mu_2$, variance $\mu_1+\mu_2$ and is supported on $\mathbb{Z}$.}
\begin{equation}
	N_a \sim \operatorname{Geom}(\tfrac{1}{3}) \,, \qquad w_a^i \sim \operatorname{Skellam}(10,10) \,.
\end{equation}
Results are insensitive to this particular choice.


\subsection{Mutations}
\label{app:mutations}
Here we give further details for the mutations $\mu_\pm$ and $\mu_{\mathfrak{S}_4}$, introduced in Sec.~\ref{sec:mutations}.

\begin{itemize}
	\item $\boldsymbol{\pm}$\textbf{:} Change a stack's winding number's signs. As discussed in Sec.~\ref{sec:symsandoperations}, there are natural `flip' operations which change the signs of $n,\widehat{m}$ rather than $n,m$. Write
	\begin{equation}
		\mu_\pm(w^i) = \big(f_1(n^1,m^1),f_2(n^2,m^2),f_3(n^3,m^3)\big) \,,
	\end{equation}
	where each $f_i$ is chosen uniformly from $\{\mathbb{I},\operatorname{flip}_n,\operatorname{flip}_{\widehat{m}},\operatorname{flip}_{n\widehat{m}}\}$.
	\item $\mathbf{\boldsymbol{\mathfrak{S}}_4}$\textbf{:} Permute a stack's $\widehat{X}^I,\widehat{Y}^I$ via the following procedure. It is helpful to think of the permutation group on $I=0,1,2,3$ as $\mathfrak{S}_4\cong\mathfrak{S}_3\times\mathbb{V}$, where $\mathbb{V}\cong\mathbb{Z}_2\times\mathbb{Z}_2$ is the Klein four-group. First, pick an element of $\mathbb{V}=\{\mathbb{I},\rho_1,\rho_2,\rho_3\}$ uniformly at random, where, for example, $\rho_1$ is given by
	\begin{equation}
		\rho_1(w^i) = \big({-n^1},{-m^1},\widehat{m}^2,{-n^2}-2b_2m^2,\widehat{m}^3,{-n^3}-2b_3m^3\big) \,,
	\end{equation}
	and with similar expressions for $\rho_2$ and $\rho_3$: the above $\rho_1$ implements the swaps $\widehat{X}^0\leftrightarrow\widehat{X}^1$, $\widehat{X}^2\leftrightarrow\widehat{X}^3$, $\widehat{Y}^0\leftrightarrow\widehat{Y}^1$ and $\widehat{Y}^2\leftrightarrow\widehat{Y}^3$. Next, pick $\sigma\in\mathfrak{S}_3$ and apply
	\begin{equation}
		(n^i,m^i) \mapsto \begin{cases}
			(n^{\sigma(i)},m^{\sigma(i)}) & b_i=b_{\sigma(i)}\\
			(n^{\sigma(i)},2m^{\sigma(i)}+n^{\sigma(i)}) & b_i=0\,,\;\; b_{\sigma(i)}=\frac{1}{2}\\
			\big(n^{\sigma(i)},\frac{m^{\sigma(i)}-n^{\sigma(i)}}{2}\big) & b_i=\frac{1}{2}\,,\;\; b_{\sigma(i)}=0
		\end{cases}
	\end{equation}
	Notice that the last case above may result in a half-integer winding number. This happens at most once, and if it does simply round up or down at random. The following are two examples applied to the same initial winding numbers:
	\begin{align}
	&\begin{aligned}
		&\!\!\!\!\underline{b_i=(0,0,0)}\\
		w^i &= (3,{-1},1,{-1},2,1) &\xrightarrow{\rho_1}\quad &({-3},1,{-1},{-1},1,{-2}) &\xrightarrow{(13)}\quad &(1,{-2},{-1},{-1},{-3},1)\\
		\widehat{X}^I &= (6,3,1,{-2}) & &(3,6,{-2},1) & &(3,1,{-2},6)\\
		\widehat{Y}^I &= (1,2,6,{-3}) & &(2,1,{-3},6) & &(2,6,{-3},1)
	\end{aligned}\\[10pt]
	&\begin{aligned}
		&\!\!\!\!\underline{b_i=(0,0,\tfrac{1}{2})}\\
		w^i &= (3,{-1},1,{-1},2,1) &\xrightarrow{\rho_2}\quad &({-1},{-3},{-1},1,4,{-3}) &\xrightarrow{(123)}\quad &({-1},1,4,{-2},{-1},-1)\\
		\widehat{X}^I &= (6,12,4,{-2}) & &(4,{-2},6,12) & &(4,6,12,{-2})\\
		\widehat{Y}^I &= (4,2,6,{-12}) & &(6,{-12},4,2) & &(6,4,2,{-12})
	\end{aligned}
	\end{align}
	Notice that in the second example above with tilted tori the winding numbers have changed value in order to keep the values of $\widehat{X}^I,\widehat{Y}^I$ the same. As a result the final values of $n^2,m^2$ are not coprime: this would be amended during the adjustment step before the child is added to the population.
\end{itemize}


\subsection{Adjustments}
\label{app:adjustments}
The final step before a child is added to the population is to adjust the chromosome to ensure that it remains within the environment. In particular, stack sizes cannot be zero or negative, winding number pairs must be coprime, and for $\env=2,3$ some brane types must be changed or removed. For each stack, if $N_a<1$ then set $N_a\to1$. If $(n,m)$ are not coprime, then there are a few cases to be corrected:

\begin{itemize}
	\item Replace $(0,0)$ with one of $(1,0),(-1,0),(0,1),(0,-1)$.
	\item Replace $(0,m)$, $|m|>1$ with one of $(1,m),(-1,m),(0,m/|m|)$.
	\item Replace $(n,0)$, $|n|>1$ with one of $(n,1),(n,-1),(n/|n|,0)$.
	\item If $n,m\neq0$ with $\gcd(n,m)\neq1$, pick either $n,m$ and increment/decrement until $\gcd(n,m)=1$. It is easy to see that one never arrives at one of the above three cases where further changes would be needed. For example:
	\begin{equation}
	\begin{aligned}
		&(n,m)=(4,2) \qquad \text{decrement }n \quad & &\Longrightarrow \qquad (4,2)\to(3,2)\\
		&(n,m)=(6,2) \qquad \text{increment }m \quad & &\Longrightarrow \qquad (6,2)\to(6,3)\to(6,4)\to(6,5)
	\end{aligned}
	\end{equation}
\end{itemize}
Finally, for $\env=2,3$ make the following sign changes:
\begin{itemize}
	\item Type $A',C'$: flip signs of all winding numbers.
	\item Type $B'$: if both nonzero $\widehat{X}^I$ are negative, then simply flip signs of all winding numbers. Otherwise, there are three cases:
	\begin{itemize}
	    \item If $\widehat{X}^0=0$, apply $\operatorname{flip}_{\widehat{m}}$ to the $n^i,m^i$ pair for which $\widehat{X}^i>0$.
	    \item If $\widehat{X}^0>0$, apply $\operatorname{flip}_{\widehat{m}}$ to an $n^i,m^i$ pair for which $\widehat{X}^i=0$.
	    \item If $\widehat{X}^0<0$, apply $\operatorname{flip}_{n}$ to an $n^i,m^i$ pair for which $\widehat{X}^i=0$.
	\end{itemize}
\end{itemize}


\section{Hyper-parameter optimization}
\label{app:optimization}

The genetic algorithm described in Sec.~\ref{sec:GA} has a number of parameters that need to be chosen to maximize its success in finding fully consistent solutions. In this appendix we fix the environment parameters to
\begin{equation}
    b_i = 0 \,, \qquad \widehat{U}_I = 1 \,, \qquad (k_\text{min},\,k_\text{max},\,N_\text{max}) = (3,\,10,\,10) \,, \qquad \env = 1 \,,
\end{equation}
and discuss the optimization of the hyperparameters. We take the quantity
\begin{equation}
    \mathcal{E} = \frac{\#\{\text{distinct TKS-solutions found in $R$ runs}\}}{R}
\end{equation}
as a measure of the genetic algorithm's efficiency. Our goal is to understand how the function $\mathcal{E}$ depends on each of the genetic algorithm's tunable parameters.

In determining the optimal choice for the hyper-parameters we choose to fix
\begin{equation}
    (n_\text{pop},\,n_\text{elite},\,g_\text{max}) \quad=\quad (250,\,25,\,100)\,.
\end{equation}
as well as $\mathcal{W}_\text{M}=0$ (so that only consistency conditions, rather than phenomenological properties, are considered). This leaves 17 adjustable hyper-parameters: eight cross-over probabilities $p_i$, four mutation rates $r_i$, three weights $\mathcal{W}_i$ and two scales $\Delta_\text{T,S}$. We make the simplifying assumption that the cross-over probabilities ``factorize'', i.e.\ that the probabilities for the eight methods may be parametrized in terms of only three probabilities, namely $p_1$, $p_\text{s}$ and $p_\text{c}$. For example,
\begin{equation}
    p_\textbf{1sc} = p_1p_\text{s}p_\text{c} \,, \qquad p_\textbf{2su} = (1-p_1)p_\text{s}(1-p_\text{c}) \,.
\end{equation}
We are then left with 11 independent hyper-parameters to be chosen, which we take to be parametrized in the following way:
\begin{equation}\label{eq:hyperparametrization}
\begin{aligned}
	h_i &\in \Big\{\,p_1,\;p_\text{s},\;p_\text{c},\; \log_{10}{(r_w)},\;\log_{10}{(r_N)},\;\log_{10}{(r_\pm)},\;\log_{10}{(r_{\mathfrak{S}_4})},\\
	&\hspace{80pt} \mathcal{W}_\text{K},\; \log_{10}\!{\big(\tfrac{\mathcal{W}_\text{T}}{\mathcal{W}_\text{S}}\big)},\; \log_{10}{(\Delta_\text{S})},\;\log_{10}\!{\big(\tfrac{\Delta_\text{T}}{\Delta_\text{S}}\tfrac{\mathcal{W}_\text{S}}{\mathcal{W}_\text{T}}\big)}\, \Big\}
\end{aligned}
\end{equation}
\begin{equation}
\begin{aligned}
    p_1, p_\text{s}, p_\text{c}, \mathcal{W}_\text{K} \in [0,1] \qquad \log_{10}{(\Delta_\text{S})} \in [-1,1] \qquad \log_{10}\!{\big(\tfrac{\Delta_\text{T}}{\Delta_\text{S}}\tfrac{\mathcal{W}_\text{S}}{\mathcal{W}_\text{T}}\big)} \in [1,3]\\
    \log_{10}{(r_w)},\log_{10}{(r_N)},\log_{10}{(r_\pm)},\log_{10}{(r_{\mathfrak{S}_4})},\log_{10}\!{\big(\tfrac{\mathcal{W}_\text{T}}{\mathcal{W}_\text{S}}\big)} \in [-1.5,0.5]\;\;\;
\end{aligned}
\end{equation}
The ratio $\mathcal{W}_\text{T}/\mathcal{W}_\text{S}$ controls the relative importance of the tadpole and SUSY terms in the fitness function when far from optimality and the ratio $(\Delta_\text{T}\mathcal{W}_\text{S})/(\Delta_\text{S}\mathcal{W}_\text{T})$ controls the relative importance of the two terms when close to optimality:
\begin{equation}
    \mathcal{F} \supset \mathcal{W}_\text{T}\;h\Big(\tfrac{\langle T^I\rangle}{\Delta_\text{T}}\Big) + \mathcal{W}_\text{S}\;h\Big(\tfrac{\langle S_a\rangle}{\Delta_\text{S}}\Big) \quad\xrightarrow[\langle S_a\rangle \lesssim \Delta_\text{S}]{\langle T^I\rangle \lesssim \Delta_\text{T}}\quad \frac{\mathcal{W}_\text{T}}{\Delta_\text{T}}\,\langle T^I\rangle + \frac{\mathcal{W}_\text{S}}{\Delta_\text{S}}\,\langle S_a\rangle \,.
\end{equation}
The values of $\mathcal{W}_\text{T}$, $\mathcal{W}_\text{S}$ and $\Delta_T$ can all be uniquely reconstructed by using $\sum_i\mathcal{W}_i=1$.

\begin{figure}[t]
    \centering
	\includegraphics[width=\textwidth]{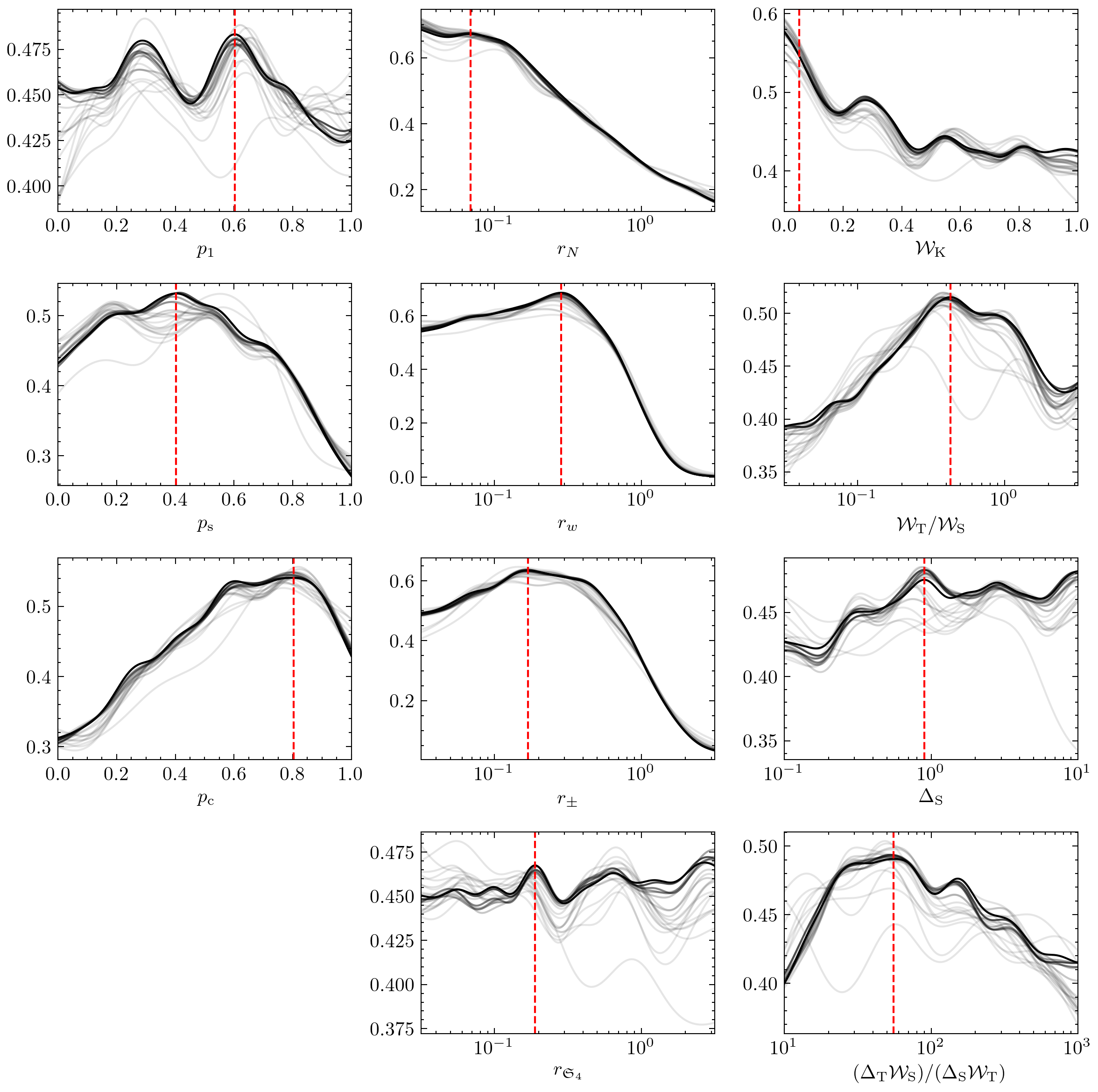}
    \caption{Efficiencies $\mathcal{E}_{h_i}(N;h)$ for the eleven tunable hyper-parameters of Eqn.~\eqref{eq:hyperparametrization}, shown for $N=500$ (lighest) to $N=10,\!000$ (darkest) in increments of $500$. The red dashed lines show the choice for each hyper-parameter.}
    \label{fig:opt}
\end{figure}

In order to efficiently sample the efficiency in this high-dimensional space we employ a random search~\cite{bergstra2012random} wherein the above set of hyper-parameters is chosen at random a number of times (labelled by $j=1,2,\ldots,N$) rather than systematically scanned over. For each set of parameters the genetic algorithm is run $R=100$ times and $\mathcal{E}^{(j)}$ computed. From these data we can construct the functions $\mathcal{E}_{h_i}(N;h)$ for each hyper-parameter $h_i$ which are computed by convolving with a Gaussian kernel $K(z)\propto\exp{(-\frac{1}{2}z^2)}$ and marginalizing over all of the other ten hyper-parameters:
\begin{equation}
    \mathcal{E}_{h_i}(N;h) = \frac{\sum_{j=1}^N \mathcal{E}^{(j)}K\big(\frac{h-h_i^{(j)}}{\sigma}\big)}{\sum_{j=1}^N K\big(\frac{h-h_i^{(j)}}{\sigma}\big)} \,, \qquad \sigma^2 = \operatorname{Var}(h_i^{(j)})\,N^{-2/5} \,.
\end{equation}
The eleven hyper-parameters of Eqn.~\eqref{eq:hyperparametrization} are sampled from uniform distributions on their respective intervals $10,\!000$ times: the resulting $\mathcal{E}_{h_i}$ are shown in Fig.~\ref{fig:opt}. We see that the efficiency is most sensitive to $p_\text{s}$, $p_\text{c}$, $r_w$, $r_N$ and $r_\pm$. Also, although the dependence on $\mathcal{W}_\text{K}$ is relatively mild, there is a preference for small values which points towards a post-selection type strategy where the K-theory condition is not heavily considered during the population's evolution. Based on these results we pick the following values (shown in red in Fig.~\ref{fig:opt}):
\begin{equation}
\begin{aligned}
    (p_1,p_\text{s},p_\text{c}) &= (0.60,\; 0.40,\; 0.80) \,,\\
    (r_w,r_N,r_\pm,r_{\mathfrak{S}_4}) &= (0.07,\; 0.28,\; 0.17,\;0.19) \,,\\
    (\mathcal{W}_\text{T},\mathcal{W}_\text{K},\mathcal{W}_\text{S}) &= (0.28,\;0.05,\;0.67) \,,\\
    (\Delta_\text{T},\Delta_\text{S}) &= (21,\; 0.90) \,.
\end{aligned}
\end{equation}
For nonzero $\mathcal{W}_\text{M}$ we simply rescale $\mathcal{W}_\text{T,K,S}\to(1-\mathcal{W}_\text{M})\mathcal{W}_\text{T,K,S}$, as in Eqn.~\eqref{eq:optparams}.

\bibliography{braneGA}

\end{document}